\documentclass[sn-mathphys]{sn-jnl}

\usepackage{mathtools}
\usepackage{lineno} 
\usepackage{comment}
\usepackage{color}
\numberwithin{equation}{section} 
\usepackage{bm}

\theoremstyle{thmstyleone}%
\newtheorem{theorem}{Theorem}[section]%
\newtheorem{proposition}[theorem]{Proposition}%
\newtheorem{corollary}[theorem]{Corollary}
\newtheorem{lemma}[theorem]{Lemma}%

\theoremstyle{thmstyletwo}%
\newtheorem{remark}[theorem]{Remark}%

\usepackage{graphicx}        
\usepackage{multicol}        
\usepackage[bottom]{footmisc}
\usepackage{newtxmath}       
\usepackage{enumerate}

\def\red#1{{#1}}

\newcommand{\be}{\begin{equation}}
\newcommand{\ee}{\end{equation}}
\newcommand{\ba}{\begin{align}}
\newcommand{\bas}{\begin{align*}}
\newcommand{\noret}{\nonumber \\ }

\newcommand{\ul}[1]{\underline{#1}}

\newcommand{\bR}{\mathbb{R}}  
\newcommand{\bN}{\mathbb{N}}
\newcommand{\cP}{\mathcal{P}} 
\newcommand{\cS}{\mathcal{S}} 

\newcommand{\tilalpha}{\tilde{\alpha}}

\newcommand{\lift}[1]{{(#1)}^\sharp}

\newcommand{\id}{\mathit{id}}

\newcommand{\Cov}{\mathrm{Cov}}

\newcommand{\vf}{\mathfrak{X}} 
\newcommand{\cD}{\mathfrak{D}} 

\newcommand{\transposedmap}[1]{{}^{\mathrm{t}}{#1}}
\newcommand{\transpose}{\transposedmap}

\newcommand{\kernel}{\mathrm{Ker}\,}

\newcommand{\any}{\forall}
\newcommand{\some}{\exists}

\newcommand{\condset}[2]{\{#1 \,\vert\, #2\}}

\newcommand{\bigcondset}[2]{\Bigl\{#1 \,\Big\vert \, #2\Bigr\}}

\newcommand{\nablastar}{\nabla^*} 

\newcommand{\nablae}{\nabla^{(\mathrm{e})}}
\newcommand{\nablaesub}[1]{\nabla^{(\mathrm{e})}_{#1}}
\newcommand{\nablam}{\nabla^{(\mathrm{m})}}
\newcommand{\nablamsub}[1]{\nabla^{(\mathrm{m})}_{#1}}
\newcommand{\rmm}{\mathrm{m}}
\newcommand{\rme}{\mathrm{e}}

\renewcommand{\tan}[1]{T_{#1}}
\newcommand{\cotan}[1]{T^*_{#1}}

\renewcommand{\tanm}[1]{T^{(\mathrm{m})}_{#1}}

\newcommand{\expect}[1]{\langle #1\rangle}
\newcommand{\inprod}[2]{\langle #1, #2\rangle}  

\newcommand{\norm}[1]{\| #1 \|} 
\newcommand{\restrict}{\vert} 

\newcommand{\subm}{submanifold}

\newcommand{\correspond}[1]{\stackrel{#1}{\longleftrightarrow}}

\newcommand{\memo}[1]  
{\bigskip \par\noindent
\fbox{\parbox[t]{12cm} 
{#1}
} \par\bigskip\noindent
}

\begin{document}

\title[The Fisher metric on the cotangent bundle]{The Fisher metric as a metric on the cotangent bundle}


\author[1]{\fnm{Hiroshi} \sur{Nagaoka}}\email{nagaoka@is.uec.ac.jp}

\affil[1]{
\orgname{The University of Electro-Communications}, \orgaddress{\street{1-5-1 Chofugaoka}, \city{Chofu, Tokyo}, \postcode{182-8585}, 
\country{Japan}}}



\abstract{
The Fisher metric on a manifold of probability distributions is usually treated  
as a metric on the tangent bundle.  In this paper, we focus on the metric on 
the cotangent bundle induced from the Fisher metric with calling it the Fisher co-metric. 
We show that the Fisher co-metric can be defined directly without going through the Fisher metric 
by establishing a natural correspondence between cotangent vectors and random variables. 
This definition clarifies a close relation between the Fisher co-metric and the variance/covariance of random variables, 
whereby the Cram\'{e}r-Rao inequality is trivialized. 
We also discuss the monotonicity and the invariance of the Fisher co-metric with respect to Markov maps, 
and present a theorem characterizing the co-metric by the invariance, which can be 
regarded as a cotangent version of 
\v{C}encov's characterization theorem for the Fisher metric. 
The obtained theorem can also \red{be} viewed as giving a characterization of the variance/covariance. 
}

\keywords{information geometry, Fisher metric, cotangent space, \v{C}encov's  (Chentsov's) theorem}



\maketitle

\section{Introduction}\label{sec_intro}

The Fisher metric on a statistical manifold (a manifold consisting of 
probability distributions) is one of the most important notions in information geometry 
\cite{amanag}. 
It is usually treated as a Riemannian metric which is a metric on the tangent bundle. 
The subject of the present paper is the metric on the cotangent bundle corresponding to the Fisher metric, 
which we call the \emph{Fisher co-metric}. The Fisher metric and the Fisher co-metric are 
essentially a single geometric object so that one is induced from another.  Nevertheless, studying the Fisher 
co-metric has several implications as mentioned below, which are what the present paper intends to show.

Firstly, as will be seen in Section~\ref{sec_Fisher_co-metric}, 
the Fisher co-metric is defined via the variance/covariance of random variables 
based on a natural correspondence between cotangent vectors and random variables. This definition 
is very natural and does not seem arbitrary. There is no room for questions such as why $\log p$ appears in the definition of the Fisher metric.  

Secondly, the above relationship between cotangent vectors and random variables 
directly links the variance/covariance of an unbiased estimator and the Fisher co-metric, 
which trivializes the Cram\'{e}r-Rao inequality. Recognizing this fact, the Fisher metric appears to be a detour 
for the Cram\'{e}r-Rao inequality, at least conceptually. 

Thirdly, once we focus on the Fisher co-metric, we are motivated to reconsider a known result for the Fisher metric as a 
source of similar problems for the Fisher co-metic and the variance/covariance, which may lead to a new insight. 
As an example,  co-metric and variance/covariance versions of \v{C}encov's theorem 
on characterization of the Fisher metric are investigated in this paper.

The paper is organized as follows. In Section~\ref{sec_Fisher_co-metric}, we introduce the Fisher co-metric 
on the manifold $\cP (\Omega)$, which is the totality of positive probability distributions on a finite set $\Omega$, 
via the variance/covariance of random variables on $\Omega$.  In Section~\ref{sec_metric_from_co-metric}, 
the Fisher co-metric is shown to be equivalent to the Fisher metric by a natural correspondence. 
In Section~\ref{sec_submanifold}, the Fisher metric and co-metric on an arbitrary {\subm} of $\cP$ are 
discussed, where we see that the Cram\'{e}r-Rao inequality is trivialized by considering the co-metric. 
 Section~\ref{sec_em_connections} 
treats the e- and m-connections on $\cP (\Omega)$, where it is clarified that, in application to 
estimation theory, the role of the m-connection as a connection on the cotangent bundle and its relation 
to the Fisher co-metric are crucial.  Sections~\ref{sec_Fisher_co-metric}-\ref{sec_em_connections} 
can be considered to constitute a first half of the paper, which is aimed at showing the naturalness and the usefulness of 
considering the Fisher co-metric.  

The second half of the paper focuses on the monotonicity and the invariance of the Fisher metric and co-metric with respect to Markov maps. 
In Section~\ref{sec_monotonicity}, we investigate the monotonicity. 
We show there that the monotonicity of the Fisher metric, which is well known as a characteristic property of the metric, 
is equivalently translated into the monotonicity of the Fisher co-metric and that of the variance. 
In Section~\ref{sec_invariance}, after reviewing the invariance of the Fisher metric and \v{C}encov's theorem, 
we consider their co-metric versions. It is shown that, being different from the monotonicity, the invariance 
of the metric and that of the co-metric are not logically equivalent. We present a theorem on characterization of the Fisher co-metric in terms of the invariance, which corresponds to  \v{C}encov's theorem but does not follow from it.  
The obtained theorem can also be expressed as a theorem on characterization of the variance/covariance. 
In Section~\ref{sec_strong_invariance}, we investigate a stronger version of the invariance, which can be 
regarded as the joint condition that combines the invariance of the metric and that of the co-metric. The formulation 
used for expressing this condition is applied to affine connections in Section~\ref{sec_invariance_connection}, 
whereby a kind of invariance condition for affine connection is obtained. 
The condition is shown to be equivalent to a known version of invariance condition which is 
seemingly weaker than the original condition used by  \v{C}encov 
to characterize the $\alpha$-connections, but actually characterizes the $\alpha$-connection as well. 
Section~\ref{sec_concluding} is devoted to concluding remarks. 

\begin{remark}
Throughout this paper, we denote the tangent space 
and the cotangent space of a manifold $M$ at a point $p\in M$ by 
$\tan{p} (M)$ and $\cotan{p} (M)$, respectively.  
We also denote the totality of smooth vector fields and that 
of smooth differential $1$-forms on $M$ by $\vf(M)$ and $\cD (M)$, respectively.  
We generally use capital letters $X, Y\red{,} \ldots $ for 
vector fields  in $\vf (M)$, which are maps assigning tangent vectors 
$X_p, Y_p, \ldots $ in $\tan{p} (M)$ to each point $p\in M$. 
To save the symbols, we also denote general tangent vectors in  $\tan{p} (M)$ 
by $X_p, Y_p, \ldots$, not only when they are the values of vector fields. 
Similarly, We use Greek letters $\alpha, \beta, \ldots $ for $1$-forms in $\cD (M)$, 
which are maps assigning cotangent vectors 
$\alpha_p, \beta_p, \ldots $ in $\cotan{p} (M)$ to each point $p\in M$, and 
also denote general cotangent vectors by $\alpha_p, \beta_p, \ldots$, not only when 
they are the values of $1$-forms.  The pairing of $X_p\in \tan{p} (M)$ and 
$\alpha_p\in \cotan{p}(M)$ is expressed as $\alpha_p (X_p)$, considering 
a cotangent vector as a function on the tangent space. 
We keep the first capital letters $A, B, \ldots$ 
for random variables ($\bR$-valued functions on sample spaces). 
\end{remark}

\section{The Fisher co-metric}
\label{sec_Fisher_co-metric}

We introduce the Fisher co-metric in this section, while its equivalence to the Fisher metric will be  
shown in the next section. 

Let $\Omega$ be a finite set with cardinality $\lvert\Omega\rvert\geq 2$, 
and let $\cP (\Omega)$ be the totality of strictly positive probability distributions on $\Omega$: 
\be
\cP=\cP (\Omega) := 
\bigcondset{\,p}{p: \Omega \rightarrow (0, 1), \; \sum_{\omega\in\Omega} p(\omega) =1 \,}, 
\label{def_P(Omega)}
\ee
which is regarded as a manifold with $\dim \cP (\Omega) = \lvert\Omega\rvert -1$. 
Let the totality of $\bR$-valued functions on $\Omega$  be denoted by $\bR^\Omega$, 
and define $\bigl(\bR^\Omega\bigr)_c := \condset{A\in\bR^\Omega}{\sum_{\omega\in\Omega} A(\omega) =c}$ 
for a constant $c\in \bR$.  Since $\cP$ is an open subset of the affine space $\bigl(\bR^\Omega\bigr)_1$, 
its tangent space can be identified with the linear space  $\bigl(\bR^\Omega\bigr)_0$. 
Following the terminology of \cite{amanag}, we denote this identification $\tan{p} (\cP) \rightarrow \bigl(\bR^\Omega\bigr)_0$ by $X_p \mapsto X_p^{(\rmm)}$, and call $X_p^{(\rmm)}$ the \emph{m-representation}  
of $X_p$. 

For an arbitrary {\subm} $M$ of $\cP$ (including the case when $M=\cP$) 
we define $\tanm{p}(M) := \condset{X_p^{(\rmm)}}{X_p \in \tan{p}(M)}$, 
which is a linear subspace of $\tanm{p}(\cP) = \bigl(\bR^\Omega\bigr)_0$.
When the elements of  $M$ are parametrized as $p_\xi$ by a coordinate system $\xi = (\xi^i)$
of $M$, the m-representation of $(\partial_i)_p \in \tan{p}(M)$, where $\partial_i:= \frac{\partial}{\partial\xi^i}$, 
 with $p=p_\xi$ is represented as 
\be
(\partial_i)_{p}^{(\rmm)} = \partial_i p_{\xi}, 
\label{partial_m-rep}
\ee
and $\{ (\partial_i)_{p}^{(\rmm)} \}_{i=1}^n$ ($n=\dim M$) constitute a basis of $\tanm{p}(M)$.

We denote the expectation of a random variable $A\in \bR^\Omega$ w.r.t. a 
distribution $p\in\cP$ by 
\begin{align}
\expect{A}_p &:= \sum_{\omega\in\Omega} p(\omega) A(\omega), 
\end{align}
and define the function 
\be
\expect{A} : \cP \rightarrow \bR, \;\; p \mapsto \expect{A}_p . 
\ee
Since $\expect{A} $ is a smooth function on the manifold $\cP$, 
its differential $(d\expect{A})_p \in \cotan{p} (\cP)$ at each point $p\in\cP$ 
is defined.  We introduce the following map: 
\be
\delta_p : \bR^\Omega \rightarrow \cotan{p} (\cP), \;\; A \mapsto \delta_p (A) :=(d\expect{A})_p \red{,}
\ee
for which we have
\be
\any A\in \bR^\Omega, \, \any X_p\in\tan{p} (\cP), \;\; 
\delta_p (A) (X_p) = X_p \expect{A} = 
\sum_{\omega\in\Omega} X_p^{(\rmm)}(\omega) A (\omega).
\label{delta(A)(X)_Xm}
\ee

\begin{proposition} 
\label{prop_d<A>p_linear_iso}
For every $p\in \cP$, the linear map $\delta_p : \bR^\Omega \rightarrow \cotan{p} (\cP)$ 
is surjective with $\kernel \delta_p =\bR$, where $\bR$ is regarded as a subspace of $\bR^\Omega$ by 
identifying a constant $c\in\bR$ with the constant function 
$\omega\mapsto c$.  Hence, $\delta_p$ induces a linear isomorphism $\bR^\Omega / \bR  \rightarrow \cotan{p} (\cP)$. 
\end{proposition}

\begin{proof}
Every cotangent vector $\alpha_p\in \cotan{p} (\cP)$ is a linear functional on $\tan{p} (\cP)$, which 
is represented as $\alpha_p : X_p \mapsto \sum_\omega X_p^{(\rmm)} (\omega ) A (\omega)$ 
by some $A\in \bR^{\Omega}$. This means that $\alpha_p = \delta_p (A)$ due to \eqref{delta(A)(X)_Xm}. 
Hence,  $\delta_p$ is surjective. For any $A\in \bR^{\Omega}$, we have
\ba
A \in \kernel \delta_p \; &\Leftrightarrow \; \any X_p\in \tan{p} (\cP ), \;\; \delta_p (A) (X_p) = 
\sum_\omega X_p^{(\rmm)} (\omega ) A (\omega) = 0
\noret &\Leftrightarrow \; \any B \in \bigl(\bR^\Omega\bigr)_0, \;\; \sum_\omega A(\omega) B(\omega) = 0
\noret &\Leftrightarrow \; A \in \bR, 
\nonumber
\end{align}
which proves $\kernel \delta_p =\bR$. 
\end{proof}

For each $p\in\cP$, denote the $L^2$ inner product and the covariance 
of random variables $A, B\in \bR^\Omega$ by 
\ba
\inprod{A}{B}_p & := \expect{AB}_p \quad \text{and}
\label{def_L2}
\\
\Cov_p (A, B) &:= \inprod{A- \expect{A}_p}{B- \expect{B}_p}_p. 
\end{align}
Then $\Cov_p : (A, B) \mapsto \Cov_p (A, B)$ 
is a degenerate nonnegative bilinear form on $\bR^\Omega$ with kernel $\bR$, and 
defines an inner product on $\bR^\Omega / \bR$.  Therefore, Proposition~\ref{prop_d<A>p_linear_iso} 
 implies that an inner product on $\cotan{p} (\cP)$, which we denote by $g_p$, 
 can be defined by 
\be
g_p (\delta_p (A), \delta_p (B)) =  \Cov_p (A, B).
\label{def_Fisher_co-metric}
\ee
Denoting the norm for $g_p$ by $\norm{\cdot}_p$, we have
\be
\norm{\delta_p(A)}_p^2 = V_p (A),
\label{d<A>2=Vp} 
\ee
where the RHS is the variance $V_p (A) :=  \expect{(A- \expect{A}_p)^2 }_p$. 

We have thus defined the map $g$ which maps each point $p\in\cP$ to the inner product $g_p$ on $\cotan{p} (\cP)$. 
We generally call such a map (a metric on the cotangent bundle) a \emph{co-metric}.
Although a co-metric is essentially equivalent to 
a usual (Riemannian) metric (a metric on the tangent bundle) by the correspondence explained in the next section, 
it is often useful to distinguish them conceptually.  The co-metric defined by \eqref{def_Fisher_co-metric} 
is called the \emph{Fisher co-metric}, since it corresponds to the Fisher metric as will be shown later. 

\begin{remark}
Eq.~\eqref{d<A>2=Vp} is found in Theorem~2.7 of the book \cite{amanag}, 
where the norm and the inner product on the cotangent space were  
considered to be induced from the Fisher metric. 
\end{remark}

 \section{The correspondence between a metric and a co-metric}
\label{sec_metric_from_co-metric}

By a standard argument of linear algebra, an inner product $\inprod{\cdot}{\cdot}$ 
on a $\bR$-linear space $V$ establishes a natural linear isomorphism between $V$ and its dual space $V^*$, 
which we denote by $\stackrel{\inprod{\cdot}{\cdot}}{\longleftrightarrow}$. 
This  gives a one-to-one correspondence between a metric on a manifold $M$ and a co-metric on $M$ as follows. 
Given a metric  $g$ on $M$, a tangent vector $X_p\in \tan{p} (M)$ and 
a cotangent vector 
$\alpha_p\in \cotan{p} (M)$ at a point $p\in M$ correspond each other by
\ba
X_p \stackrel{g_p}{\longleftrightarrow} \alpha_p 
& \; \Leftrightarrow\;
\any Y_p\in \tan{p} (M), \; \alpha_p (Y_p) = g_p (X_p, Y_p).
\label{X_gp_alpha_1}
\end{align}
The correspondence is extended to the correspondence between a vector field $X\in \vf (M)$ and 
a 1-form $\alpha\in \cD (M)$ by 
\be
X \stackrel{g}{\longleftrightarrow} \alpha \; \Leftrightarrow \; \any p\in M, \;\; X_p \stackrel{g_p}{\longleftrightarrow} \alpha_p. 
\ee
(Note: some literature refers to this correspondence as the musical isomorphism with notation $\alpha = X^\flat$
and $X = \alpha^\sharp$, while we will use the symbol $\sharp$ for a different meaning later.)
This correspondence determines a co-metric on $M$,  which is denoted by the same symbol $g$, 
such that for every $p\in M$ 
\be
X_p \stackrel{g_p}{\longleftrightarrow} \alpha_p \;\; 
\text{and} \;\;
Y_p \stackrel{g_p}{\longleftrightarrow} \beta_p 
\; \Rightarrow \;
g_p (\alpha_p, \beta_p) = g_p (X_p, Y_p). 
\label{gp_XY_alphabeta}
\ee
Conversely,  given a co-metric $g$ on $M$, the correspondence $\stackrel{g_p}{\longleftrightarrow}$ is defined by
\ba
X_p \stackrel{g_p}{\longleftrightarrow} \alpha_p 
& \; \Leftrightarrow\;
\any \beta_p\in \cotan{p} (M), \; \beta_p (X_p) = g_p (\alpha_p, \beta_p), 
\label{X_gp_alpha_2}
\end{align}
and a metric on $M$ is defined by the same relation as \eqref{gp_XY_alphabeta}. 
It should be noted that when a metric and a co-metric correspond in this way, 
the relations 
\eqref{X_gp_alpha_1} and \eqref{X_gp_alpha_2} are equivalent, so that 
there arises no confusion even if we use the same symbol $g$ for 
the corresponding metric and co-metric in $\stackrel{g_p}{\longleftrightarrow}$ and $\stackrel{g}{\longleftrightarrow}$. 

Note that for an arbitrary coordinate system $(\xi^i)$ of $M$, $g_{ij} := g(\frac{\partial}{\partial\xi^i}, \frac{\partial}{\partial\xi^j})$
and $g^{ij} := g(d\xi^i, d\xi^j)$ form the inverse matrices of each other at every point of $M$. 
Note also that the norms for $(\tan{p}(M), g_p)$ and $(\cotan{p}(M), g_p)$ 
are linked by
\ba
\norm{X_p}_{p} & = \max_{\alpha_p \in \cotan{p} (M)\setminus \{0\}} \frac{\lvert\alpha_p (X_p)\rvert}{\norm{\alpha_p}_{p}} 
\label{norm_max_tangent_manifold}
\\
\text{and}\qquad 
\norm{\alpha_p}_{p} & = \max_{X_p\in \tan{p} (M)\setminus \{0\}} \frac{\lvert\alpha_p (X_p)\rvert}{\norm{X_p}_{p}}, 
\label{norm_max_cotangent_manifold}
\end{align}
where the $\max$'s in these equations are achieved by those $X_p$ and $\alpha_p$ 
which correspond to each other by $\correspond{g_p}$ 
up to a constant factor. 

For a tangent vector $X_p\in\tan{p} (\cP)$, define 
\be
L_{X_p} := X_p^{(\rmm)} / p\;  \in 
\condset{A\in \bR^\Omega}{\expect{A}_\red{p} =0} , 
\label{def_LX}
\ee
which is the derivative of the map $\cP \rightarrow \bR^\Omega$,  $p\mapsto \log p$ w.r.t.\ $X_p$. 
(In \cite{amanag}, $L_{X_p}$ is called the \emph{e-representation} of $X_p$ and is denoted by $X_p^{(\rme)}$.) 
Note that $L_{X_p}$ is characterized by (cf.\ \eqref{delta(A)(X)_Xm}) 
\be
\any A\in \bR^{\Omega}, \; \delta_p (A) (X_p) = 
X_p \expect{A} = \inprod{L_{X_p}}{A}_p.
\label{delta(A)(X)_LX}
\ee

The following proposition shows that the metric induced from the Fisher co-metric $g$ 
by the correspondence $\stackrel{g}{\longleftrightarrow}$ is the Fisher metric. 

\begin{proposition}
\label{prop_metric_tangent_cotangent}
For each point $p\in \cP$, we have:
\begin{enumerate}
\item $\any A\in \bR^{\Omega}, \;\any X_p\in \tan{p}(\cP), \;\; 
X_p  \stackrel{g_p}{\longleftrightarrow} \delta_p (A) \; \Leftrightarrow\; L_{X_p} = A - \expect{A}_p$. 
\item $\any X_p, Y_p\in \tan{p}(\cP), \;\; g_p (X_p, Y_p) = \inprod{L_{X_p}}{L_{Y_p}}$. 
\end{enumerate}
\end{proposition}

\begin{proof} 
\ul{1}: According to \eqref{X_gp_alpha_2}, the condition $X_p  \stackrel{g_p}{\longleftrightarrow} \delta_p (A)$ is equivalent to
\be
\any B\in \bR^\Omega, \; g_p (\delta_p (A), \delta_p (B)) = \delta_p (B) (X_p). 
\label{cond_g(deltaAdeltaB)}
\ee
Here the LHS is equal to 
\[
\inprod{A-\expect{A}_p}{B-\expect{B}_p}_p = \inprod{A-\expect{A}_p}{B}_p, 
\]
while the RHS is equal to $\inprod{L_{X_p}}{B}_p$ by \eqref{delta(A)(X)_LX}. Hence, 
 \eqref{cond_g(deltaAdeltaB)} is equivalent to $L_{X_p} = A-\expect{A}_p$. 
 \\
 \ul{2}: Obvious from item 1 and \eqref{gp_XY_alphabeta}. 
\end{proof}

\section{The Fisher co-metric on a submanifold and the Cram\'{e}r-Rao inequality}
\label{sec_submanifold}

Let $M$ be an arbitrary submanifold of $\cP$. 
Then a metric on $M$ is induced as the restriction of the Fisher metric $g$, which we denote by $g_M : p\mapsto g_{M, p}
= g_p\restrict_{\tan{p}(M)^2}$. 
When a coordinate system $\xi = (\xi^i)$ is given on $M$, corresponding to 
\eqref{partial_m-rep} it holds that
\be
L_{(\partial_i)_p} = \partial_i \log p_\xi \quad\text{at} \quad p = p_\xi.
\label{partial_LX}
\ee
We have
\be
g_{M, ij} (p) := g_{M, p} ((\partial_i)_p, (\partial_j)_p) = \inprod{ \partial_i \log p_\xi}{ \partial_j \log p_\xi}_p, 
\ee
which defines the Fisher information matrix $G_M (p) = [g_{M, ij} (p) ]$. 
The metric $g_{M}$ induces a co-metric on $M$, which is denoted by the same symbol $g_{M}$. 
Letting 
\be
g_{M}^{ij} (p) := g_{M, p} ((d\xi^i)_p, (d\xi^j)_p), 
\label{g^ij_M_p}
\ee
we have  $G_M (p)^{-1} =  [g_{M}^{ij} (p) ]$. 

Suppose that a cotangent vector $\alpha_p \in \cotan{p}(M)$ on $M$ is the restriction of a cotangent vector 
$\tilde{\alpha}_p \in  \cotan{p}(\cP)$ on $\cP$; i.e., $\alpha_p = \tilde{\alpha}_p\restrict_{\tan{p} (M)}$. 
Then, it follows from \eqref{norm_max_cotangent_manifold} that
\ba
\norm{\alpha_p}_{M, p} 
& \red{
=  \max_{X_p\in \tan{p} (M)\setminus \{0\}} \frac{\lvert\alpha_p (X_p)\rvert}{\norm{X_p}_{M, p}}
}
\noret 
&=  \max_{X_p\in \tan{p} (M)\setminus \{0\}} \frac{\lvert\tilde{\alpha}_p (X_p)\rvert}{\norm{X_p}_{p}}
\noret 
& \leq  \max_{X_p\in \tan{p} (\cP)\setminus \{0\}} \frac{\lvert\tilde{\alpha}_p (X_p)\rvert}{\norm{X_p}_{p}} 
= \norm{\tilde{\alpha}_p}_{p} .
\end{align}
(Note that $\norm{X_p}_p = \norm{X_p}_{M, p}$ since the metric on $M$ is the restriction of 
the metric on $\cP$.) 
Furthermore, for an arbitrary $\alpha_p \in \cotan{p}(M)$, there always exists $\tilde{\alpha}_p \in\cotan{p} (\cP)$ 
satisfying $\alpha_p = \tilde{\alpha}_p\restrict_{\tan{p} (M)}$ and 
$\norm{\alpha_p}_{M, p} =\norm{\tilde{\alpha}_p}_{p}$.  Indeed, letting $X_p\in \tan{p}(M)$ be defined by $X_p \correspond{g_{M, p}} \alpha_p$, such an $\tilde{\alpha}_p$ is obtained by $X_p \correspond{g_{p}} \tilde{\alpha}_p$.  

The above observations lead to the following proposition.

\begin{proposition}
\label{prop_norm_submanifold} \
\begin{enumerate}
\item For any $\alpha_p \in \cotan{p} (M)$, we have
\ba
\norm{\alpha_p}_{M, p} & = \min \condset{\norm{\tilalpha_p}_p\,}{\,\tilalpha_p\in \cotan{p}(\cP) \;\;\text{and}\;\; \alpha_p = \tilalpha_p\restrict_{\tan{p}(M)}} 
\noret 
& = \norm{\lift{\alpha_p}}_p, 
\label{eq1_prop_norm_submanifold}
\end{align}
where $\lift{\alpha_p} := \arg\min_{\tilalpha_p}\condset{\norm{\tilalpha_p}_p}{\cdots}$. 
\item For any  $\alpha_p, \beta_p \in \cotan{p} (M)$, we have
\be
g_{M, p} (\alpha_p, \beta_p) = g_p (\lift{\alpha_p}, \lift{\beta_p}).
\label{eq2_prop_norm_submanifold}
\ee
\end{enumerate}
\end{proposition}

The above proposition shows that the Fisher co-metric on $M$ 
can be defined from the Fisher co-metric on $\cP$ directly
by \eqref{eq1_prop_norm_submanifold} and \eqref{eq2_prop_norm_submanifold}, not by way of the Fisher metric.

\begin{corollary}[The Cram\'{e}r-Rao inequality]
\label{cor_Cramer-Rao}
Suppose that an $n$-tuple of random variables $\vec{A}=(A^1 , \dots \red{,}\,  A^n )\in (\bR^\Omega)^n$ satisfies
\be
\any i\in\{1, \ldots , n\}, \;\; (d\xi^i)_p = \delta_p (A^i)\restrict_{\tan{p}(M)}
\label{locally_unbiased_1}
\ee
for a coordinate system $\xi = (\xi^i)$ of an $n$-dimensional {\subm} $M$ of $\cP$ and for a point  
$p\in M$. Letting  $V_p (\vec{A}) = [v^{ij}] \in \bR^{n\times n}$ be the variance-covariance matrix of  $\vec{A}$ defined by
\be
v^{ij}  := \Cov_p (A^i, A^j)
\label{def_vij_classical}
\ee
and letting $G_M (p)$ be the Fisher information matrix, 
we have
\be
V_p (\vec{A}) \geq {G_M (p)}^{-1}.
\label{Cramer-Rao_classical}
\ee
\end{corollary}

\begin{proof} For an arbitrary column vector $c= (c_i)\in \bR^n$, let 
\ba
\tilalpha_p &:= \sum_i c_i \delta_p (A^i) \in \cotan{p} (\cP), 
\noret
\alpha_p &:= \sum_i c_i (d\xi^i)_p \in \cotan{p} (M). 
\nonumber
\end{align}
Since \eqref{locally_unbiased_1} implies that $\tilalpha_p\restrict_{\tan{p}(M)} = \alpha_p$, 
it follows from Prop.~\ref{prop_norm_submanifold} that $\norm{\tilalpha_p}_p \geq \norm{\alpha_p}_{M, p}$. 
Noting that $\norm{\tilalpha_p}_p^2 = \transpose{c}\, V_p (\vec{A})\,  c$ and $\norm{\alpha_p}_{\red{M, p}}^2 = \transpose{c}\, {G_M (p)}^{-1} c$, where $\transpose{}$ denotes the transpose, 
we obtain \eqref{Cramer-Rao_classical}. 
\end{proof}

\section{On the e, m-connections}
\label{sec_em_connections}

An affine connection is usually treated as a connection on the tangent bundle, 
while it corresponds to a connection on the cotangent bundle by the relation
\be
\any X, Y \in \vf (M), \; \any \alpha\in \cD (M), \;\; X \alpha (Y) = \alpha (\nabla_X Y) +( \nabla_X \alpha ) (Y).
\label{eq_Xalph(Y)}
\ee
This correspondence is one-to-one, so that we can define an affine connection by specifying a connection 
on the cotangent bundle.  Therefore,  the $\alpha$-connection in information geometry can also be introduced in this way. 
Although affine connections are out of the main subject of this paper, we will briefly discuss the significance of defining the m-connection (i.e.\ $(\alpha = -1)$-connection) in this way, since it is closely related to the role of the Fisher co-metric 
in the Cram\'{e}r-Rao inequality. 

We start by introducing the m-connection $\nablam$ on $\cP = \cP (\Omega)$ as 
a flat connection on the cotangent bundle 
for which the 1-form $d\expect{A}$ is parallel for any $A\in \bR^\Omega$; i.e., 
\be
\any X\in \vf (\cP), \any A\in \bR^\Omega, \; \nablam_X d\expect{A} =0.
\label{m-connection_cotangent_P}
\ee
\red{Since  $\dim \condset{d \expect{A}}{A\in \bR^\Omega} = \vert \Omega\vert - 1 = \dim \cP$, 
\eqref{m-connection_cotangent_P} implies that 
every parallel 1-from is represented as $d\expect{A}$ by some $A\in \bR^\Omega$. 
}
Then the correspondence \eqref{eq_Xalph(Y)} determines a connection on 
the tangent bundle, which is denoted by the same symbol $\nablam$. 
Letting $\alpha = d\expect{A}$ in  \eqref{eq_Xalph(Y)}  and applying 
\eqref{m-connection_cotangent_P}, we have
\be
\any X, Y\in \vf (\cP), \any A\in \bR^\Omega, \; 
X Y \expect{A} = (\nablamsub{X} Y) \expect{A} .
\ee
This implies that, for any $Y\in \vf (\cP)$, 
\ba
& \text{$Y$ is m-parallel} 
\noret 
& \; \Leftrightarrow \; \any A\in \bR^\Omega, \, \any X\in  \vf (\cP), \; X Y \expect{A} =0 
\noret 
& \; \Leftrightarrow \; \any A\in \bR^\Omega, \; 
\text{
$Y_p \expect{A} = \sum_\omega Y_p^{(\rmm)} (\omega) A(\omega)$  does not depend on $p\in\cP$
}
\noret 
& \; \Leftrightarrow \;
\text{$Y_p^{(\rmm)}$ does not depend on $p\in\cP$}, 
\end{align}
where ``m-parallel'' means ``parallel w.r.t.\ $\nablam$''. 
Since this property characterizes the m-connection on $\cP$ (e.g. 
Eq.~(2.39) of \cite{amanag}), our definition of the m-connection 
is equivalent to the usual definition in information geometry. 

Next,  we define the e-connection $\nablae$ as the dual connection of $\nablam$ w.r.t.\ the Fisher metric $g$ 
(\cite{amanag}, \cite{metr}), which means that 
\be
\any X, Y, Z\in \vf (\cP), \; Z g(X, Y) = g(\nablaesub{Z} X, Y) + g(X, \nablamsub{Z} Y).
\label{eq_dual_em_classical}
\ee
Using \eqref{eq_Xalph(Y)}, we can rewrite \eqref{eq_dual_em_classical} into 
\ba
\any X, Y, Z  \in & \vf (\cP),\, \any \alpha \in\cD (\cP), \;
\noret 
& 
\,X \correspond{g} \alpha \; \Rightarrow\; 
(\nablamsub{Z} \alpha) (Y) = g (\nablaesub{Z} X, Y).
\end{align}
This  implies that, for any $X\in \vf (\cP)$ and $\alpha\in\cD (\cP)$, 
\be
X \correspond{g} \alpha \; \Rightarrow \; 
\Bigl[\,
\text{$X$ is e-parallel} \; \Leftrightarrow \; \text{$\alpha$ is m-parallel}
\, \Bigr]. 
\label{X_eparallel_alpha_mparallel}
\ee

Now, let us recall the situation of Corollary~\ref{cor_Cramer-Rao}.  An estimator $\vec{A}=(A^1 , \dots \red{,} \, A^n )$ is said to be \emph{efficient} for the statistical model $(M, \xi)$ when it is unbiased (i.e.\ $\any i$,  $\xi^i = \expect{A^i}\restrict_{M}$) 
and achieves the equality in the Cram\'{e}r-Rao inequality \eqref{Cramer-Rao_classical} for every $p\in M$. 
Noting that the achievability at each $p\in M$ is represented by the condition $\any i$, $\delta_p (A^i) 
\red{= (d\expect{A^i})_p} 
= \lift{(d\xi^i)_p}$ 
and recalling \eqref{m-connection_cotangent_P}, 
we can see that the condition for $(M, \xi)$ to have an efficient estimator is expressed as
\be
\any i, \; \some \tilalpha^i \in \cD (\cP), \;\; \text{$\tilalpha^i$ is m-parallel and}\;\; \any p\in M, \; \tilalpha^i_p = \lift{(d\xi^i)_p}. 
\label{cond_efficient_geometrical}
\ee
On the other hand, it is well known that the existence of an efficient estimator is equivalent to the 
condition that $M$ is an exponential family 
and that $\xi$ is an expectation coordinate system, which can be rephrased as (see Theorem~3.12 of \cite{amanag}) 
\begin{gather}
\text{$M$ is an e-autoparallel submanifold of $\cP$}, 
\noret 
\text{and $\xi$ is an an m-affine coordinate system.}
\label{cond_e-family_m-affine}
\end{gather}

Therefore, the two conditions \eqref{cond_efficient_geometrical} and \eqref{cond_e-family_m-affine} 
are necessarily equivalent.  These  are  both purely geometrical conditions for 
a submanifold of the dually flat space $(\cP, g, \nablae, \nablam)$, and we can prove their equivalence 
within this geometrical framework, forgetting its statistical background.  Indeed, 
the equivalence can be proved for a more general situation where $M$ is  a submanifold of 
a manifold $S$ equipped with a Riemannian metric $g$ and a pair of dual affine connections 
$\nabla, \nablastar$ on the assumption that 
$\nablastar$ is flat. Note that this assumption is weaker than the dually-flatness 
of $(S, g, \nabla, \nablastar)$ 
in that $\nabla$ is allowed to have non-vanishing torsion, which is essential 
in application to quantum estimation theory.  See section~7 of \cite{nagfuji_autoparallel} for details.

\section{Monotonicity}
\label{sec_monotonicity}

The monotonicity with respect to a Markov map is known to be an important and characteristic property of the Fisher metric. 
In this section we discuss the monotonicity of the Fisher co-metric and its relation to the variance of random variables. 

Let $\Omega_1$ and $\Omega_2$ be arbitrary finite sets, and let $\cP_i:= \cP (\Omega_i)$ for $i=1, 2$.
A map $\Phi : \cP_1 \rightarrow \cP_2$ is called a \emph{Markov map} 
when it is affine in the sense that $\any p, q\in \cP_1 $, $0\leq \any a \leq 1$, $\Phi (a p + (1-a) q ) = a\, \Phi (p) + (1-a)\, \Phi (q)$.  Every Markov map $\Phi$ is represented as
\be
\any p\in \cP_1, \;\; 
\Phi (p) = \sum_{x\in \Omega_1} W(\cdot \,\vert\, x) p(x), 
\label{Phi_W}
\ee
where $W$ is a surjective channel from $\Omega_1$ to $\Omega_2$; i.e.,
\begin{gather}
\any (x, y)\in \Omega_1\times \Omega_2, \; W(y\,\vert \, x) \geq 0, \quad 
\any x\in \Omega_1, \; \sum_{y\in \Omega_2} W(y\,\vert \, x) =1, \\
\text{and}\quad 
\any y\in \Omega_2, \; \some x\in \Omega_1, \; W(y\,\vert \, x) >0.
\end{gather}
When $\Phi$ is represented as \eqref{Phi_W}, we write $\Phi = \Phi_W$. 

More generally, for a {\subm} $M$ of $\cP_1$ and a {\subm} $N$ of $\cP_2$, a map $\varphi : M \rightarrow N$ 
is called a \emph{Markov map} when there exists a Markov map $\Phi : \cP_1 \rightarrow \cP_2$ such that $\varphi = \Phi\restrict_M$. 
Since a Markov map $\varphi$ is smooth, it induces at each $p\in M$ 
the differential
\be
\varphi_* = \varphi_{*,p} = (d\varphi)_p : \tan{p} (M) \rightarrow \tan{\varphi (p)} (N)
\ee
and its dual
\be
\varphi^* = \varphi^*_p = \transposedmap{ (d\varphi)_p } : \cotan{\varphi(p)} (N) \rightarrow \cotan{p} (M),
\ee
where $\transposedmap{}$ denotes the transpose of a linear map. 
See Remark~\ref{rem_phi^*} below for the notation $\varphi^* = \varphi^*_p$. 

As is well known, the Fisher metric satisfies the following monotonicity property for its norm: 
\be
\any p\in M, \any X_p\in  \tan{p} (M), \;\; 
\norm{\varphi_*(X_p)}_{\varphi (p)} \leq \norm{X_p}_{p}.
\label{monotonicity_tangent_classical}
\ee
The cotangent version of the monotonicity is given below.

\begin{proposition} We have
\label{prop_monotonicity_cotangent_classical}
\be
\any p\in M, \any \alpha_{\varphi(p)} \in \cotan{\varphi (p)} (N), \;\; 
\norm{\varphi^*(\alpha_{\varphi(p)})}_{M, p} \leq \norm{\alpha_{\varphi(p)}}_{N, \varphi(p)}, 
\label{monotonicity_cotangent_classical}
\ee
where  $\norm{\cdot}_{M, p}$ and $\norm{\cdot}_{N, \varphi (p)}$ denote the norms w.r.t.\  
the Fisher co-metrics $g_{M}$ and $g_{N}$, respectively.  
\end{proposition}

\begin{proof} \, 
\red{
Since the inequality is trivial when  $\norm{\varphi^*(\alpha_{\varphi(p)})}_{M, p} = 0$, we 
assume $\norm{\varphi^*(\alpha_{\varphi(p)})}_{M, p} > 0$. 
}
Then, invoking \eqref{norm_max_cotangent_manifold},  we have
\ba
\norm{\varphi^*(\alpha_{\varphi(p)})}_{M, p} 
& = 
\max_{X_p\in \tan{p} (M) \setminus \{0\}} \frac{\vert \varphi^* (\alpha_{\varphi (p)}) (X_p)\vert}{\norm{X_p}_{\red{M, p}}}
\noret 
&= \max_{X_p\in \tan{p} (M) \setminus \{0\}} \frac{\vert \alpha_{\varphi (p)} (\varphi_*(X_p))\vert}{\norm{X_p}_p}
\noret 
&\red{= \max_{X_p\in \tan{p} (M) : \varphi_*(X_p) \neq 0} \frac{\vert \alpha_{\varphi (p)} (\varphi_*(X_p))\vert}{\norm{X_p}_p}}
\noret 
&\leq \red{\sup_{X_p\in \tan{p} (M) : \varphi_*(X_p) \neq 0}}
 \frac{\vert \alpha_{\varphi (p)} (\varphi_*(X_p))\vert}{\norm{\varphi_*(X_p)}_{\varphi (p)}}
\noret 
&\leq \max_{Y_{\varphi (p)}\in \tan{\varphi(p)} (N) \setminus \{0\}} \frac{\vert \alpha_{\varphi (p)} (Y_{\varphi(p)})\vert}{\norm{Y_{\varphi(p)}}_{\varphi(p)}}
=  \norm{\alpha_{\varphi (p)}}_{N, \varphi (p)}, 
\nonumber
\end{align}
where \red{the third equality follows since $X_p$ achieving $\max_{X_p\in \tan{p} (M) \setminus \{0\}}$ 
should satisfy 
$\varphi_*(X_p) \neq 0$ due to the assumption  $\norm{\varphi^*(\alpha_{\varphi(p)})}_{M, p} > 0$, and }
the first $\leq$ follows from \eqref{monotonicity_tangent_classical}. 
\end{proof}

\begin{remark}
\label{rem_phi^*}
We have written $\varphi^* (\alpha_{\varphi(p)})$ for $\varphi^*_p (\alpha_{\varphi(p)})$ above 
(and will use similar notations throughout the paper), considering that 
omitting $p$ from $\varphi^*_p$ is harmless in the context and is better for the readability of expressions. 
Note that the notation $\varphi^* (\alpha_{\varphi(p)})$, if it appears alone,  is mathematically ambiguous 
\red{unless $\varphi^{-1} (\varphi (p))$ is the singleton $\{p\}$ } in that 
$\varphi^*_{p'} (\alpha_{\varphi(p)}) \in \cotan{p'}(M)$ \red{depends on a choice of $p' \in\varphi^{-1} (\varphi (p))$}. 
On the other hand, \red{the notation} $\varphi_*(X_p)$ has no such ambiguity, 
since \red{we know that the argument $X_p$ belongs to $\tan{p}(M)$ and hence  
$\varphi_*$ must be $\varphi_{*, p}$}.
\end{remark}

Let us consider the case when $M = \cP_1$ and $N = \cP_2$, and let 
$\Phi = \Phi_W: \cP_1 \rightarrow \cP_2$ be an arbitrary Markov map 
represented by a surjective channel $W$. 
Recalling \eqref{Phi_W} and the definition of m-representation of tangent vectors, we have 
\be
Y_{\Phi (p)} = \Phi_* (X_p )  \; \Leftrightarrow \; Y_{\Phi (p)}^{(\rmm)} =
 \sum_{x\in \Omega_1} W(\cdot \,\vert\, x) X_p^{(\rmm)} (x) 
 \label{m-rep_YPhi(p)}
\ee
for $X_p \in \tan{p}(\cP_1)$ and $Y_{\Phi (p)}\in \tan{\Phi (p)} (\cP_2)$. 
We claim that 
\be
\any A\in \bR^{\Omega_2}, \;\; 
\Phi^* (\delta_{\Phi (p)} (A)) = \delta_p (E_W (A\,\vert\,\cdot)), 
 \label{Phi*_delta_Phi(p)_A}
\ee
where $\Phi^* = \Phi^*_p$, and 
$E_W (A\,\vert\,\cdot) \in \bR^{\Omega_1}$ denotes the conditional expectation of $A$ defined by
\be
\any x\in \Omega_1, \;\;
E_W (A\,\vert\,x) = \sum_{y\in \Omega_2} W(y \,\vert\, x) A(y).
\ee
Eq.\ \eqref{Phi*_delta_Phi(p)_A} is verified as follows; 
for every $\beta_{p} = \delta_p (B) \in \cotan{p} (\cP_1)$, where 
$B \in \bR^{\Omega_2}$, we have
\ba
& \beta_{p} = \Phi^* (\delta_{\Phi (p)} (A)) 
\noret & \; \Leftrightarrow \; 
\any X_p\in \tan{p} (\cP_1), \;\; \beta_{p} (X_p) =\delta_{\Phi (p)} (A)(\Phi_* (X_p)) 
\noret & \; \Leftrightarrow \; 
\any X_p\in \tan{p} (\cP_1), \;\; 
\sum_{x\in \Omega_1} X^{(\rmm)}_p (x) B(x) = 
\sum_{(x, y) \in \Omega_1\times \Omega_2} W(y \,\vert\, x) X_p^{(\rmm)} (x) A(y) 
\noret & \; \Leftrightarrow \; 
B - E_W (A\,\vert\,\cdot) \in \bR
\noret & \; \Leftrightarrow \; 
\beta_p = \delta_p ( E_W (A\,\vert\,\cdot)), 
 \label{A_B_Phi^*}
\end{align}
where the second $\Leftrightarrow$ follows from 
\eqref{delta(A)(X)_Xm} and \eqref{m-rep_YPhi(p)}, the third $\Leftrightarrow$ 
follows from  $\tanm{p}(\cP) = \bigl(\bR^\Omega\bigr)_0$, 
and
$\bR$ is identified with the set of constant functions on $\Omega_1$. 

Invoking \eqref{d<A>2=Vp} and \eqref{Phi*_delta_Phi(p)_A}, we see that the monotonicity \eqref{monotonicity_cotangent_classical} is 
equivalent to  the following well-known inequality for the variance: 
\be
\any A\in \bR^{\Omega_2}, \;\; V_{p} (E_W (A\,\vert\,\cdot) ) \leq 
V_{\Phi(p)} (A), 
\label{monotonicity_variance}
\ee
which we refer to as the \emph{monotonicity of the variance}. 

In the above proof of Prop.~\ref{prop_monotonicity_cotangent_classical}, 
we derived \eqref{monotonicity_cotangent_classical} from \eqref{monotonicity_tangent_classical}. 
Conversely, we can derive \eqref{monotonicity_tangent_classical} from \eqref{monotonicity_cotangent_classical} 
by the use of \eqref{norm_max_tangent_manifold} as follows; for any $X_p\in \tan{p}(M)$, 
\ba
\norm{\varphi_*(X_p)}_{\varphi (p)} 
&= \max_{\alpha_{\varphi(p)}\in\cotan{\varphi(p)}(N)\setminus \{0\}}
\frac{\vert \alpha_{\varphi(p)} (\varphi_* (X_p))\vert}{\norm{\alpha_{\varphi(p)}}_{N, \varphi(p)}}
\noret
& = \max_{\alpha_{\varphi(p)}\in\cotan{\varphi(p)}(N)\setminus \{0\}}
\frac{\vert \varphi^*(\alpha_{\varphi(p)}) (X_p)\vert}{\norm{\alpha_{\varphi(p)}}_{N, \varphi(p)}}
\noret
& \red{= \max_{\alpha_{\varphi(p)}\in\cotan{\varphi(p)}(N): 
\varphi^*(\alpha_{\varphi(p)}) \neq 0}
\frac{\vert \varphi^*(\alpha_{\varphi(p)}) (X_p)\vert}{\norm{\alpha_{\varphi(p)}}_{N, \varphi(p)}}
}
\noret
& \leq 
\red{
\sup_{\alpha_{\varphi(p)}\in\cotan{\varphi(p)}(N): 
\varphi^*(\alpha_{\varphi(p)}) \neq 0}
}
\frac{\vert \varphi^*(\alpha_{\varphi(p)}) (X_p)\vert}
{\norm{
\varphi^*(\alpha_{\varphi(p)})
}_{M, p}}
\noret
& \leq \max_{\beta_{p}\in\cotan{p}(M)\setminus \{0\}}
\frac{\vert \beta_{p} (X_p)\vert}
{\norm{
\beta_{p}
}_{M, p}} 
= \norm{X_p}_p,
\end{align}
where 
\red{we have assumed $\norm{\varphi_*(X_p)}_{\varphi (p)} >0$ with no loss of generality, 
which yields the third equality (cf.\  the proof of Prop.~\ref{prop_monotonicity_cotangent_classical}), and 
}
the first $\leq$ follows from  \eqref{monotonicity_cotangent_classical}.  Thus, 
 \eqref{monotonicity_tangent_classical} and \eqref{monotonicity_cotangent_classical} are equivalent. 
Note that this equivalence is derived solely from a general argument on metrics and co-metrics, 
and does not rely on the special characteristics of the Fisher metric/co-metric. In this sense, we say that 
 \eqref{monotonicity_tangent_classical} and \eqref{monotonicity_cotangent_classical} 
 are \emph{logically equivalent}. 
 
 Recalling that the Fisher metric is characterized as the unique monotone metric up to a constant factor, we obtain the following propositions from the logical equivalence mentioned above.

 \begin{proposition}
\label{proposition_characterization_co-metric_monotonicity}
 The monotonicity \eqref{monotonicity_cotangent_classical} 
 characterizes the Fisher co-metric up to a constant factor. 
 \end{proposition}

\begin{proposition}
The variance is characterized up to a constant factor as the positive quadratic form for random variables satisfying the monotonicity \eqref{monotonicity_variance}. 
\end{proposition}

\begin{remark}
We have described the above propositions in a rough form for the sake of readability. 
For the exact statement, we need a formulation similar to Theorems~\ref{theorem_invariance_metric}, 
\ref{theorem_invariance_co-metric} and \ref{theorem_invariance_covariance} in the next section. 
See also Remark~\ref{remark_characterization_co-metric_monotonicity_strong_invariance}. 
\end{remark}

\begin{remark}
Since the monotonicity of the Fisher metric  \eqref{monotonicity_tangent_classical}, that of the Fisher co-metric 
 \eqref{monotonicity_cotangent_classical}, 
and that of the variance \eqref{monotonicity_variance} are all logically equivalent, 
we can derive  \eqref{monotonicity_tangent_classical} from the more popular \eqref{monotonicity_variance}. 
\end{remark}

\section{Invariance}
\label{sec_invariance}

\v{C}encov showed in \cite{cencov} that the Fisher metric is characterized up to a constant factor 
as a covariant tensor field of degree 2 satisfying the invariance for Markov embeddings. 
Note that the invariance is weaker than the 
monotonicity and that the tensor field is not assumed to be positive nor symmetric. 
 In this section we review \v{C}encov's theorem and then investigate its co-metric version, which will be 
 shown to be equivalent to a theorem characterizing the variance/covariance of random variables. 

We begin by reviewing the invariance property of the Fisher metric. Suppose that 
$M$ and $N$ are arbitrary {\subm}s of $\cP_1 = \cP (\Omega_1)$ and $\cP_2 = \cP (\Omega_2)$, respectively, 
and that a pair of Markov maps 
\be
\varphi : M \rightarrow N \quad\text{and}\quad \psi : N \rightarrow M
\label{varphi_psi_given}
\ee
satisfis
\be
\psi \circ \varphi = \id_M, 
\label{psi_varphi_id}
\ee
where $\circ$ denotes the composition of maps. 
Note that $\varphi$ is injective while $\psi$ is surjective. 
Given a pair of points $(p, q) \in M\times N$ satisfying
\be
q = \varphi (p) \quad\text{and}\quad p = \psi (q),
\label{q_p_varphi_psi}
\ee
we have
\be
\psi_{*, q} \circ \varphi_{*, p} = \id_{\tan{p}(M)}. 
\label{psi*_varphi*_id_T}
\ee
It then follows from the monotonicity \eqref{monotonicity_tangent_classical} that
\be
\norm{X_p}_{M, p} \geq \norm{\varphi_* (X_p)}_{N, q} \geq \norm{\psi_* (\varphi_* (X_p))}_{M, p} = \norm{X_p}_{M, p},
\ee
so that we have the invariance of the Fisher metric  
\be
\any X_p\in \tan{p} (M), \;\; 
\norm{X_p}_{M, p} =  \norm{\varphi_* (X_p)}_{N, q}, 
\label{invariance_Fisher_tangent_1}
\ee
which is equivalent to
\be
\any X_p, \any Y_p\in \tan{p} (M), \;\; 
g_{M, p} (X_p, Y_p) = g_{N, q} (\varphi_* (X_p), \varphi_* (Y_p)) \red{.}
\label{invariance_Fisher_tangent_1var}
\ee
This means that $\varphi_{*, p} : \tan{p}(M) \rightarrow \tan{q} (N)$ is isometry, which is represented as
\be
(\varphi_{*, p})^\dagger \circ \varphi_{*, p} = \id_{\tan{p}(M)}, 
\label{invariance_Fisher_tangent_2}
\ee
where $(\varphi_{*, p})^\dagger : \tan{q}(N) \rightarrow \tan{p} (M)$ denotes the adjoint (Hermitian conjugate) of 
$\varphi_{*, p} $ w.r.t.\ the inner products $g_{M, p}$ and $g_{N, q}$. 

A \red{Markov} map $\Phi : \cP_1\rightarrow \cP_2$ is called an \emph{Markov embedding} when 
there exists a Markov map $\Psi :  \cP_2\rightarrow \cP_1$ such that 
\be
\Psi \circ \Phi = \id_{\cP_1}. 
\label{Psi_Phi_id}
\ee
Note that $\vert\Omega_1\vert \leq \vert\Omega_2\vert$ necessarily holds in this case. 
As a special case of the invariance \eqref{invariance_Fisher_tangent_1var}, we have
\be
\any p\in \cP_1,\, \any X_p, \any Y_p\in \tan{p} (\cP_1), \;\; 
g_p (X_p, Y_p) = g_{\Phi (p)} (\Phi_* (X_p), \Phi_* (Y_p)).
\ee
According to \v{C}encov, this property characterizes the Fisher metric up to a constant factor. 
The exact statement is presented below.

\begin{theorem}[\v{C}encov \cite{cencov}]
\label{theorem_invariance_metric}
For $n=2, 3, \ldots$, let $\Omega_n:= \{1, 2, \ldots , n\}$ and $\cP_n:= \cP (\Omega_n)$, and let 
$g_n$ be the Fisher metric on $\cP_n$.  Suppose that we are given a sequence $\{h_n\}_{n=2}^\infty$, where 
$h_n$ is a covariant tensor field of degree 2 on $\cP_n$ which  continuously maps each point $p\in \cP_n$ to a bilinear form $h_{n, p} : \tan{p} (\cP_n)^2 \rightarrow \bR$.  Then the following two conditions are equivalent.
\begin{enumerate}
\item[(i)] $\some c\in \bR, \, \any n, \;\; h_n = c g_n$.
\item[(ii)] For any $m\leq n$ and any Markov embedding $\Phi : \cP_m \rightarrow \cP_n$, it holds that
\ba
\any p\in \cP_m,\, & \any X_{p}, \any Y_p \in \tan{p} (\cP_m), 
\noret & 
h_{m, p} (X_{p}, Y_{p}) = h_{n, \Phi (p)} (\Phi_* (X_p),  \Phi_* (Y_p)).
\label{(ii)_theorem_invariance_metric}
\end{align}
\end{enumerate} 
\end{theorem}

We now proceed to the invariance property of co-metrics. Let us consider the same situation as 
\eqref{varphi_psi_given}-\eqref{q_p_varphi_psi}, which implies that
\be
\varphi_p^*\circ\psi_q^* = \id_{\cotan{p}(M)}. 
\ee
Then it follows from the monotonicity \eqref{monotonicity_cotangent_classical} that, 
 for any $\alpha_{p}\in\cotan{p}(M)$, 
\be
\norm{\alpha_p}_{M, p} \geq \norm{\psi^* (\alpha_p)}_{N, q} \geq \norm{\varphi^*(\psi^*  (\alpha_p))}_{M, p} = \norm{\alpha_p}_{M, p},
\ee
so that we have the invariance of the Fisher co-metric 
\be
\any \alpha_p \in \cotan{p}(M), \;\; 
\norm{\alpha_p}_{M, p} =  \norm{\psi^* (\alpha_p)}_{N, q}, 
\label{invariance_Fisher_cotangent_1}
\ee
which is equivalent to 
\be
\any \alpha_p, \beta_p \in \cotan{p}(M), \;\; 
g_{M, p}(\alpha_p, \beta_p) =  g_{N, q} (\psi^* (\alpha_p), \psi^* (\beta_p))
\label{invariance_Fisher_cotangent_1var}
\ee
This means that $\psi^*_{q} : \cotan{p}(M) \rightarrow \cotan{q} (N)$ is isometry, which is represented as
\be
(\psi^*_{q} )^\dagger \circ \psi^*_{q}  = \id_{\cotan{p}(M)}, 
\label{invariance_Fisher_cotangent_2}
\ee
where $(\psi^*_{q} )^\dagger  : \cotan{q}(N) \rightarrow \cotan{p} (M)$ denotes the adjoint (Hermitian conjugate) of 
$\psi^*_{q}$ w.r.t.\ the inner products $g_{M, p}$ and $g_{N, q}$ on the cotangent spaces.  
Due to $\psi^*_{q} = \transposedmap{(\psi_{*, q})}$,  \eqref{invariance_Fisher_cotangent_2} can 
be rewritten as
\be
\psi_{*, q} \circ (\psi_{*, q})^\dagger =  \id_{\tan{p}(M)}. 
\label{invariance_Fisher_cotangent_3}
\ee 

We observed in Section~\ref{sec_monotonicity} that the monotonicity of metrics \eqref{monotonicity_tangent_classical}  and that of co-metrics \eqref{monotonicity_cotangent_classical}  are 
logically equivalent.  
On the other hand, such an equivalence does not hold for the invariance properties \eqref{invariance_Fisher_tangent_1} and \eqref{invariance_Fisher_cotangent_1}. 
Indeed, a linear-algebraic consideration shows that the implications 
\red{
\eqref{psi*_varphi*_id_T} $\wedge$  \eqref{invariance_Fisher_tangent_2} $\Rightarrow$  \eqref{invariance_Fisher_cotangent_3}} and \red{\eqref{psi*_varphi*_id_T} $\wedge$  \eqref{invariance_Fisher_cotangent_3} $\Rightarrow$  \eqref{invariance_Fisher_tangent_2}} do not hold  
unless  $\varphi_{*, p}$ is a linear isomorphism (cf.\ Lemma~\ref{lemma_for_strong_invariance}). 
This means that we cannot expect \v{C}encov's theorem to yield a corollary which states that the Fisher co-metric 
is characterized by the invariance \eqref{invariance_Fisher_cotangent_1}.  Nevertheless, the statement itself is 
true as explained below. 

Let us return to the situation of \eqref{Psi_Phi_id}.  We call a Markov map $\Psi : \cP_2 \rightarrow \cP_1$ 
a \emph{Markov co-embedding} when there exists a Markov embedding $\Phi: \cP_1 
\rightarrow \cP_2$ satisfying \eqref{Psi_Phi_id}.  
As an example of \eqref{invariance_Fisher_cotangent_1}, 
\eqref{Psi_Phi_id} implies the invariance
\be
\any p\in \cP_1, \any \alpha_p\in\cotan{p} (\cP_1), \;\; 
 \norm{\alpha_p}_p = \norm{\Psi^* (\alpha_p)}_{\Phi (p)}, 
  \label{Invariance_cotangent_Psi_Phi_1}
\ee
which can be rewritten as
\be
\any q\in \Phi (\cP_1), \any \alpha_{\Psi (q)} \in\cotan{\Psi (q)} (\cP_1), \;\; 
 \norm{\alpha_{\Psi (q)}}_{\Psi (q)} = \norm{\Psi^* (\alpha_{\Psi (q)})}_{q}.
 \label{Invariance_cotangent_Psi_Phi_2}
\ee
Actually, the range $\any q\in \Phi (\cP_1)$ in the above equation can be extended to  $\any q\in \cP_2$ 
for the reason described below.  

It is known (e.g.\ Lemma~9.5 of  \cite{cencov})  that every pair $(\Phi, \Psi)$ of Markov embedding and co-embedding 
satisfying \eqref{Psi_Phi_id} is represented in the following form: 
\be
\any q \in \cP_2, \;\; \Psi (q) = q^F \quad\text{with}\quad q^F (x) := \sum_{y\in F^{-1} (x)} q(y), 
\label{Psi_F}
\ee
and
\be
\any p \in \cP_1, \;\; \Phi (p) = \sum_{x\in \Omega_1} p(x)\, r_x , 
\label{Phi_F_rx}
\ee
where 
$F$ is a surjection $\Omega_2\rightarrow \Omega_1$ which yields the partition 
$\Omega_2 = \bigsqcup_{x\in \Omega_1} F^{-1} (x)$, and $\{r_x\}_{x\in \Omega_1}$ 
is a family of probability distributions on $\Omega_2$ such that the support of 
$r_x$ is $F^{-1}(x)$ for every $x\in \Omega_1$. 
We note that a Markov co-embedding $\Psi$ is determined by $F$ alone, while 
a Markov embedding $\Phi$ is determined by $F$ and $\{r_x\}_{x\in \Omega_1}$ together. 
Consequently, $\Psi$ is uniquely determined from $\Phi$, while $\Phi$ for a given $\Psi$ 
has the degree of freedom corresponding to $\{r_x\}_{x\in \Omega_1}$. 
According to this fact, 
when a Markov co-embedding $\Psi$ and a distribution $q\in \cP_2$ are arbitrarily given, 
we can always choose a Markov embedding $\Phi$ satisfying \eqref{Psi_Phi_id} and 
$q \in \Phi (\cP_1)$; indeed, defining $r_x$ by 
\be
r_x (y) := \left\{
\begin{array}{ll}
q(y) / q^F (x) & \text{if} \;\; F(y) = x \\
0 & \text{otherwise}, 
\end{array}\right. 
\label{def_rx(y)}
\ee
the resulting $\Phi$ satisfies $q = \Phi (q^F) \in \Phi (\cP_1)$. 
This is the reason why $\any q\in \Phi (\cP_1)$ in \eqref{Invariance_cotangent_Psi_Phi_2} can be replaced with  $\any q\in \cP_2$.
We thus have 
\be
\any q\in \cP_2, \any \alpha_{\Psi (q)} \in\cotan{\Psi (q)} (\cP_1), \;\; 
 \norm{\alpha_{\Psi (q)}}_{\Psi (q)} = \norm{\Psi^* (\alpha_{\Psi (q)})}_{q}, 
 \label{Invariance_cotangent_Psi}
\ee
or equivalently, 
\ba
\any q\in \cP_2,  \any \alpha_{\Psi (q)},\;  &\any  \beta_{\Psi (q)} \in\cotan{\Psi (q)} (\cP_1), \;\; 
\noret 
& 
g_{\Psi (q)} (\alpha_{\Psi (q)}, \beta_{\Psi (q)}) = 
g_{q} (\red{\Psi^*} (\alpha_{\Psi (q)}), \red{\Psi^*}( \beta_{\Psi (q)})) 
 \label{Invariance_cotangent_Psi_var}
\end{align}
for every Markov co-embedding $\Psi$. 

The invariance \eqref{Invariance_cotangent_Psi} characterizes the Fisher co-metric up to a constant factor. Namely, we have the following theorem.

\begin{theorem}
\label{theorem_invariance_co-metric}
For $n=2, 3, \ldots$, let $\Omega_n:= \{1, 2, \ldots , n\}$ and $\cP_n:= \cP (\Omega_n)$, and let 
$g_n$ be the Fisher co-metric on $\cP_n$.  Suppose that we are given a sequence $\{h_n\}_{n=2}^\infty$, where 
$h_n$ is a contravariant tensor field of degree 2 on $\cP_n$ which  continuously maps each point $p\in \cP_n$ to a bilinear form $h_{n, p} : \cotan{p} (\cP_n)^2 \rightarrow \bR$.  Then the following two conditions are equivalent.
\begin{enumerate}
\item[(i)] $\some c\in \bR, \, \any n, \;\; h_n = c g_n$.
\item[(ii)] For any $m\leq n$ and any Markov co-embedding $\Psi : \cP_n \rightarrow \cP_m$, it holds that
\ba
\any q\in \cP_n, \, \any \alpha_{\Psi (q)}, & 
\any \beta_{\Psi (q)}  \in \cotan{\Psi (q)} (\cP_m), 
\noret & 
h_{m, \Psi (q)} (\alpha_{\Psi (q)}, \beta_{\Psi (q)}) = h_{n, q} (\Psi^* (\alpha_{\Psi (q)}), \Psi^* (\beta_{\Psi (q)})).
\label{(ii)_theorem_invariance_co-metric}
\end{align}
\end{enumerate} 
\end{theorem}

The proof will be given by rewriting the statement in terms of variance/covariance for random variables. 
Suppose that a Markov co-embedding $\Psi : \cP_1\rightarrow \cP_2$ is represented as \eqref{Psi_F} by a surjection $F: \Omega_2 \rightarrow \Omega_1$. 
Then $\Psi$ is represented as $\Psi = \Phi_W$ by the channel $W$ 
from $\Omega_2$ to $\Omega_1$ defined by 
\be
W (x \,\vert\, y) = \left\{ \begin{array}{cl} 1 &\;\text{if}\;\; x = F(y), 
\\ 0 & \; \text{otherwise}.
\end{array} \right. 
\ee
For an arbitrary $A\in \bR^{\Omega_1}$, its conditional expectation w.r.t.~$W$ is represented as
\be
E_W (A\,\vert\, y) = \sum_{x\in \Omega_1} W(x\,\vert\, y) A(x) = A (F(y)), 
\ee
so that it follows from \eqref{Phi*_delta_Phi(p)_A} that 
\be
\Psi^* (\delta_{q^F} (A) ) = \delta_q (A\circ F), 
\label{Psi*_delta_A_F}
\ee
where we have invoked $\Psi (q) = q^F$ from \eqref{Psi_F}. 
Hence, \eqref{Invariance_cotangent_Psi} and  \eqref{Invariance_cotangent_Psi_var} 
are rewritten as
\be
V_{q^F} (A) = V_q (A\circ F)
\label{invariance_variance}
\ee
and
\be
\Cov_{q^F} (A, B) = \Cov_q (A\circ F, B\circ F).
\label{invariance_covariance}
\ee
These identities themselves are obvious, but what is important is 
that 
they characterize the variance/covariance up to a constant factor.  
Namely, we have the following theorem. 

\begin{theorem}
\label{theorem_invariance_covariance}
In the same situation as Theorem~\ref{theorem_invariance_co-metric}, 
suppose that we are given a sequence $\{\gamma_n\}_{n=2}^\infty$, where 
$\gamma_n$ is a map which continuously maps each point $p\in \cP_n$ to a bilinear form $\gamma_{n, p}$ 
on $\bR^{\Omega_n}$. 
Then the following conditions (i) and (ii) are equivalent.
\begin{enumerate}
\item[(i)] $\some c\in \bR, \, \any n,  \any p\in \cP_n, 
\;\; \gamma_{n, p}  = c\, \mathrm{Cov}_p$. 
\item[(ii)] : (ii-1) $\wedge$ (ii-2)
\begin{itemize}
\item[(ii-1)] $ \any n, \any p\in \cP_n, \any A\in \bR^{\Omega_n}, \;\;  
\gamma_{n, p} (A, 1) =0$. 
\item[(ii-2)] 
For any $m\leq n$ and any surjection $F: \Omega_n \rightarrow \Omega_m$, 
it holds that 
\be
\any p\in \cP_n, \any A, B \in \bR^{\Omega_m},  \;\; 
\gamma_{m, p^F} (A, B) = \gamma_{n, p} (A\circ F, B\circ F).
\label{invariance_gamma_1}
\ee
\end{itemize}
\end{enumerate} 
\end{theorem}

Note that, if we assume that $\{\gamma_n\}_n$ are all symmetric tensors, then 
\eqref{invariance_gamma_1} can be replaced with
\be
\any p\in \cP_n, \any A \in \bR^{\Omega_m},  \;\; 
\gamma_{m, p^F} (A, A) = \gamma_{n, p} (A\circ F, A\circ F), 
\label{invariance_gamma_2}
\ee
which corresponds to \eqref{invariance_variance}. 

See \ref{subsec_proof_theorem_invariance_covariance} in Appendix for the proof, 
where we use an argument similar to 
\v{C}encov's proof of Theorem~\ref{theorem_invariance_metric}. 
It is obvious that 
Theorem~\ref{theorem_invariance_co-metric} 
immediately follows from this theorem. 

\begin{remark}
\label{remark_theorem_invariance_L2}
If we delete (ii-1) from (ii) in Theorem~\ref{theorem_invariance_covariance}, 
then we have (i)${}'$ $\Leftrightarrow $ (ii-2)  by replacing (i) with 
\begin{enumerate}
\item[(i)${}'$]
$
\some c_1, \some c_2\in \bR, \, \any n,  \any p\in \cP_n,$ 
\\ \qquad 
$ \any A, B\in \bR^{\Omega_n}, 
\;\; \gamma_{n, p}(A, B)  = c_1\, \inprod{A}{B}_p + c_2\, \expect{A}_p \expect{B}_p.
$
\end{enumerate}
We give a proof for  (i)${}'$ $\Leftrightarrow $ (ii-2) in \ref{subsec_proof_theorem_invariance_covariance}, 
from which Theorem~\ref{theorem_invariance_covariance} is straightforwad. 
\end{remark}

\section{Strong invariance}
\label{sec_strong_invariance}

In the preceding two sections, we have observed the following facts.
\begin{itemize}
\item The monotonicity of metrics and that of co-metrics are logically equivalent.
\item The monotonicity logically implies the invariance of metrics and that of co-metrics.
\item The invariance of metrics and that of co-metrics are not logically equivalent.
\end{itemize}
In this section we introduce a new notion of invariance called the strong invariance, 
and show that: 
\begin{itemize}
\item The strong invariance of metrics and that of co-metrics are logically equivalent.
\item The monotonicity of metrics/co-metrics logically implies the strong invariance of metrics/co-metrics.
\item The strong invariance of metrics/co-metrics logically implies the 
invariance of metrics and that of co-metrics. 
\end{itemize}

Recall the situation of \eqref{varphi_psi_given}, \eqref{psi_varphi_id} and \eqref{q_p_varphi_psi}, where 
we are given 
submanifolds $M \subset \cP_1$ and $N \subset \cP_2$, Markov mappings 
$\varphi : M \rightarrow N$ and $\psi : N \rightarrow M$ satisfying $\psi \circ \varphi = \id_M$, 
and points $p\in M$ and $q\in N$ satisfying $q = \varphi (p)$ and $p = \psi (q)$.  
Then we have the following proposition. 

\begin{proposition} 
\label{prop_strong_invariance_metric}
The Fisher metric\red{s} $g_M$ and $g_N$ on $M$ and $N$ satisfy 
\be
\any X_p\in \tan{p} (M), \any Y_q\in \tan{q} (N), \;\; 
g_{M, p} (X_p, \psi_* (Y_q)) = g_{N, q} (\varphi_* (X_p), Y_q), 
\label{strong_invariance_Fisher_tangent_1}
\ee
or equivalently, 
\be
\psi_{*, q} = (\varphi_{*, p})^\dagger.
\label{strong_invariance_Fisher_tangent_1var}
\ee
In addition, $(\psi_{*, q})^\dagger \circ \psi_{*, q} = \varphi_{*, p} \circ \psi_{*, q}$ 
 is the orthogonal projector from 
$\tan{q} (N)$ onto $\varphi_{*, p} (\tan{p} (M)) 
= \tan{q} (\varphi (M))$. 
\end{proposition}

The property \eqref{strong_invariance_Fisher_tangent_1} is called 
the \emph{strong invariance} of the Fisher metric. 
The proposition will be proved by using the following lemma. 

\begin{lemma} 
\label{lemma_for_strong_invariance}
Let $U$ and $V$ be finite-dimensional metric linear spaces, 
and let $A : U\rightarrow V$ and $B: V\rightarrow U$ be linear maps satisfying $B A = I$. Then the following two conditions are equivalent.
\begin{enumerate}
\item[(i)] $A^\dagger A = B B^\dagger = I$. 
\item[(ii)] $B = A^\dagger$.
\end{enumerate}
When these conditions hold, $B^\dagger B = A B$ is the orthogonal projector from $V$ 
onto the image $\mathrm{Im}\, A$ of $A$.
\end{lemma}

\begin{proof}
It is obvious that (ii) $\Rightarrow$ (i) under the assumption $B A = I$. 
\red{
Conversely, if we assume (i) with $B A = I$, then we have
\ba (B- A^\dagger)  (B - A^\dagger)^\dagger &= B B^\dagger - B A - A^\dagger B^\dagger + A^\dagger A 
\noret &= I - I  - I + I = 0, 
\nonumber 
\end{align}
from which (ii) follows.
}

Assume (i) and (ii). Then $(B^\dagger B)^2 = B^\dagger B$ due to $B B^\dagger =I$, which 
implies that $B^\dagger B$ is the orthogonal projector onto  $\mathrm{Im}\, B^\dagger 
= \mathrm{Im}\, A$. 
\end{proof}

\noindent \red{{\it Proof of Prop.~\!\ref{prop_strong_invariance_metric}} \; 
Letting} 
$U := \tan{p} (M)$, $V :=\tan{q} (\red{N})$, $A := \varphi_{*, p}$ and $B:= \psi_{*,q}$ 
in the previous lemma, we see from \eqref{psi*_varphi*_id_T}, \eqref{invariance_Fisher_tangent_2} and \eqref{invariance_Fisher_cotangent_3} that the assumption $BA = I$ and the condition (i) are satisfied, 
so that we have (ii), which means \eqref{strong_invariance_Fisher_tangent_1var}.
\qed

\medskip

As can be seen from Lemma~\ref{lemma_for_strong_invariance}, 
the strong invariance  \eqref{strong_invariance_Fisher_tangent_1} is equivalent to the condition that 
 both the invariance for the metric  \eqref{invariance_Fisher_tangent_1}-\eqref{invariance_Fisher_tangent_2} and the invariance for the co-metric \eqref{invariance_Fisher_cotangent_1}-\eqref{invariance_Fisher_cotangent_2} hold. 
Taking the transpose of both sides of \eqref{strong_invariance_Fisher_tangent_1var}, the strong 
invariance is also expressed 
\be
\psi_{q}^* = (\varphi_{p}^*)^\dagger, 
\label{strong_invariance_Fisher_cotangent_1}
\ee
which means that the Fisher co-metric satisfies
\be
\any \alpha_p\in \cotan{p} (M), \any \beta_q\in \cotan{q} (N), \;\; 
g_{M, p} (\alpha_p, \varphi^* (\beta_q)) = g_{N, q} (\psi^* (\alpha_p), \beta_q).
\label{strong_invariance_Fisher_cotangent_2}
\ee

Since the strong invariance \eqref{strong_invariance_Fisher_cotangent_2} logically implies the 
invariance \eqref{invariance_Fisher_tangent_1var} for the Fisher metric via \eqref{strong_invariance_Fisher_tangent_1}, 
we see that the following proposition, which is stated in a rough form similar to Prop.~\ref{proposition_characterization_co-metric_monotonicity}, is obtained as a corollary of \v{C}encov's theorem (Theorem~\ref{theorem_invariance_metric}). 

\begin{proposition}
 \label{proposition_characterization_co-metric_strong_invariance}
 The strong invariance \eqref{strong_invariance_Fisher_cotangent_2}  characterizes the Fisher co-metric up to a constant factor. 
\end{proposition}

\red{
An exact formulation of this proposition will be given in  \ref{subsec_exact_proposition_characterization_co-metric_strong_invariance} of Appendix with a proof based on \v{C}encov's theorem.
}

\begin{remark}
\label{remark_characterization_co-metric_monotonicity_strong_invariance}
The above proposition is stronger than Prop.~\ref{proposition_characterization_co-metric_monotonicity} and weaker than 
Theorem~\ref{theorem_invariance_co-metric}.  
To formulate Prop.~\ref{proposition_characterization_co-metric_monotonicity} and Prop.~\ref{proposition_characterization_co-metric_strong_invariance} 
in exact forms similar to Theorem~\ref{theorem_invariance_co-metric}, it matters what assumptions  
should be imposed on bilinear forms $\{h_{n, p}\}$ on the cotangent spaces prior to the monotonicity or the strong invariance. 
Here we should keep in mind that the significance of these propositions, which are weaker than Theorem~\ref{theorem_invariance_co-metric}, lies in the fact that they follow from \v{C}encov's theorem 
while Theorem~\ref{theorem_invariance_co-metric} does not. 
For Prop.~\ref{proposition_characterization_co-metric_monotonicity}, we need to assume that 
$\{h_{n, p}\}$ are inner products (positive symmetric forms) to ensure that the monotonicity of them is translated 
into the monotonicity of the corresponding inner products on the tangent spaces. For Prop.~\ref{proposition_characterization_co-metric_strong_invariance}, on the other hand, 
we only need 
to assume $\{h_{n, p}\}$ to be non-degenerate (nonsingular) and symmetric. 
See \ref{subsec_exact_proposition_characterization_co-metric_strong_invariance}  for details. 
\end{remark}

Let us consider the strong invariance \eqref{strong_invariance_Fisher_cotangent_2} for the case when $(\phi, \psi)$ is a Markov embedding/co-embedding pair $(\Phi, \Psi)$ 
and rewrite it into an identity for the covariance of random variables. 
 Let $\Omega_1$ and $\Omega_2$ be arbitrary 
finite sets satisfying $\vert \Omega_1 \vert \leq \vert \Omega_2 \vert$ and let $\cP_i := \cP (\Omega_i)$, $i=1, 2$. 
Given a surjection $F : \Omega_2 \rightarrow \Omega_1$ and a distribution $q\in \cP_2$, let 
$(\Phi, \Psi)$ be defined by \eqref{Psi_F} and \eqref{Phi_F_rx} with  \eqref{def_rx(y)}. 
For arbitrary $A \in \bR^{\Omega_1}$ and $B \in \bR^{\Omega_2}$, 
let 
$\alpha_{q^F} := \delta_{q^F} (A) \in \cotan{{q^F}} (\cP_1)$ and 
$\beta_q := \delta_q (B) \in \cotan{q} (\cP_2)$, for which 
the strong invariance \eqref{strong_invariance_Fisher_cotangent_2} is represented as
\be
g_{q^F} (\alpha_{q^F}, \Phi^* (\beta_q)) = g_q (\Psi^* (\alpha_{q^F}), \beta_q).
\label{strong_invariance_Fisher_cotangent_Phi_Psi}
\ee
Recalling \eqref{Psi*_delta_A_F}, we have
\be
\Psi^* (\alpha_{q^F}) = \delta_q (A\circ F).
\ee
On the other hand, $\Phi$ is represented as $\Phi = \Phi_V$ by 
the channel $V$ defined by
\be
V (y\,\vert\, x) = \left\{\begin{array}{ll}
r_x (y) = \frac{q(y)}{q^F (x)}&\; \text{if}\;\; x = F(y) \\
0 &\; \text{otherwise}, 
\end{array}\right.
\ee
so that \eqref{Phi*_delta_Phi(p)_A} 
yields
\be
\Phi^* (\beta_q) = \delta_{q^F} (E_V (B\, \vert\, \cdot )).
\ee
Hence, the strong invariance  \eqref{strong_invariance_Fisher_cotangent_Phi_Psi} is 
rewritten as
\be
\Cov_{q^F} (A, E_V (B\,\vert\, \cdot)) = \Cov_q (A\circ F, B).
\label{strong_invariance_covariance} 
\ee
This identity can be verified directly as follows. 
Letting $a := \expect{A}_{q^F} = \expect{A\circ F}_q$ and 
$b:= \expect{B}_q = \expect{E_V (B\, \vert\, \cdot)}_{q^F}$, we have
\ba
\text{RHS} & = \sum_y q(y) (A (F(y)) - a) (B(y) -b)
\noret &= 
\sum_x \sum_{y\in F^{-1} (x)} q(y) (A (F(y)) - a) (B(y) -b) 
\noret &= \sum_x \,  (A (x) - a) \sum_{y\in F^{-1} (x)} q(y) (B(y) -b) 
\noret &= \sum_x q^F (x) (A (x) - a)   (E_V (B\,\vert\, x) - b) 
\noret &=\text{LHS} , 
\end{align}
where the fourth equality follows from 
\ba
E_V (B\, \vert\, x) = \sum_y V (y\,\vert\, x) B(y) 
&= \frac{1}{q^F (x)} 
\sum_{y\in F^{-1} (x)} q (y) B(y).
\end{align}
Note that \eqref{invariance_covariance} is obtained from 
\eqref{strong_invariance_covariance} 
by substituting $B\circ F$ for $B$. 

\section{Weak invariance for affine connections}
\label{sec_invariance_connection}

In addition to characterizing the Fisher metric by the invariance with respect to 
Markov embeddings, \v{C}encov also 
gave a characterization of the $\alpha$-connections by the invariance condition. 
In this section we show that a similar notion to the strong invariance of metrics, which is described in terms of Markov embedding/co-embedding pairs,  can be considered for affine connections. 

Let $\Omega_1, \Omega_2$ be arbitrary finite sets satisfying $2\leq \vert \Omega_1 \vert \leq \vert \Omega_2 \vert$, 
and let $\Phi : \cP_1 \rightarrow \cP_2$ be a Markov embedding, where $\cP_i := \cP (\Omega_i)$, $i=1, 2$. 
Suppose that affine connections 
$\nabla$ and $\nabla'$ are given on $\cP_1$ and $\cP_2$, respectively.
When these connections are the $\alpha$-connection on $\cP_1$ and that on $\cP_2$ 
for some common $\alpha\in \bR$, they satisfy
\be
\any X, Y \in \vf (\cP_1), \; \; \Phi_* (\nabla_X Y) = \nabla'_{\Phi_* (X)} \Phi_* (Y).
\label{invariance_connection}
\ee
Some remarks on the meaning of the above equation are in order. First, 
we define $\Phi_* (X)$ for an arbitrary vector field $X$ on $\cS_1$ as 
a vector field on $K:=  \Phi (\cP_1)$ that maps each point $q=\Phi (p)\in K$, where $p\in \cP_1$, to 
\be
(\Phi_* (X))_q := \Phi_{*, p} (X_p) \in \tan{q} (K).
\ee
Since $K$ is a {\subm} of $\cS_2$ on which the connection $\nabla'$ is given, 
 $\nabla'_{\Phi_* (X)} \Phi_* (Y)$ in 
\eqref{invariance_connection} is defined as a map which maps each point 
$q\in K$ to a tangent vector in $\tan{q}(\cP_2)$, although  $\nabla'_{\Phi_* (X)} \Phi_* (Y)$ 
does not belong to $\vf (K)$ in general. 
The condition \eqref{invariance_connection} means that $K$ is autoparallel in $\cP_2$ with respect to
$\nabla'$ and that the restricted connection of $\nabla'$ induced on the autoparallel $K$ 
is obtained from $\nabla$ by the diffeomorphism $\Phi : \cP_1 \rightarrow K$. 
\v{C}encov \cite{cencov} showed that the invariance condition characterizes 
the family $\{\text{$\alpha$-connection}\}_{\alpha\in \bR}$ by a  formulation similar 
to Theorem~\ref{theorem_invariance_metric}. 

\begin{remark}
As is mentioned above, 
the fact that $\{\text{$\alpha$-connection}\}_{\alpha\in \bR}$ satisfy the invariance \eqref{invariance_connection} 
implies that $K = \Phi (\cP_1)$ is autoparallel in $\cP_2$ w.r.t.\ the $\alpha$-connection 
for every $\alpha\in \bR$. A kind of converse result is found in \cite{nagaoka_isit2017}, which states that 
if a submanifold $K$ of $\cP_2 = \cP (\Omega_2)$ is autoparallel w.r.t.\ the $\alpha$-connection 
for every $\alpha\in \bR$ (or, for some two different values of $\alpha$), then 
$K$ is represented as $\Phi (\cS_1)$ by some Markov embedding $\Phi$ from some $\cS_1 = \cS (\Omega_1)$ 
into $\cS_2$. 
\end{remark}

Let $g$ and $g'$ be the Fisher metrics on $\cP_1$ and $\cP_2$, respectively. 
Noting that $\Phi_*$ is an isometry with respect to these metrics due to the 
invariance of the Fisher metric, 
\eqref{invariance_connection} implies that
\be
\any X, Y, Z\in \vf (\cS_1), \;\; g(\nabla_X Y , Z) = g' (\nabla'_{\Phi_* (X)} \Phi_* (Y), \, \Phi_* (Z)) \circ \Phi, 
\label{weak_invariance_connection_1}
\ee
where the RHS denotes the function on $\cP_1$ such that
\be
p \mapsto g'_{\Phi (p)} \bigl( \bigl( \nabla'_{\Phi_* (X)} \Phi_* (Y)\bigr)_{\Phi (p)}, \, 
\Phi_{*, p}  (Z_p) \bigr).
\ee
Since \eqref{weak_invariance_connection_1} is apparently weaker than 
\eqref{invariance_connection}, we call this property the \emph{weak invariance} 
of the connections $\nabla, \nabla'$. 
In an actual fact, however, 
as is mentioned in \cite{fujiwara_hommage} and is proved in 
\cite{fujiwara_japanese}, the weak invariance \eqref{weak_invariance_connection_1} 
characterizes 
the family $\{\text{$\alpha$-connection}\}_{\alpha\in \bR}$ as well as the stronger condition \eqref{invariance_connection}. 

Now, recalling the strong invariance \eqref{strong_invariance_Fisher_tangent_1} of the Fisher metric, we see that 
the weak invariance \eqref{weak_invariance_connection_1} is equivalent to
\be
\any X, Y \in \vf (\cS_1), \;\; 
 \nabla_X Y = \Psi_* \bigl( \nabla'_{\Phi_* (X)} \Phi_* (Y) \bigr), 
 \label{weak_invariance_connection_2}
\ee
where the RHS denotes the map
\be
\cP_1 \ni p \mapsto \Psi_{*, \Phi (p)}  \bigl(  \bigl( \nabla'_{\Phi_* (X)} \Phi_* (Y) \bigr)_{\Phi (p)} \bigr).
\ee
The fact that the weak invariance is a condition on connections is more clearly expressed in \eqref{weak_invariance_connection_2} than \eqref{weak_invariance_connection_1} 
in that \eqref{weak_invariance_connection_2} does not include the Fisher metric. 

\section{Concluding remarks}
\label{sec_concluding}

In this paper we have focused on the face of the Fisher metric as a metric on the cotangent bundle, calling it the Fisher co-metric to distinguish it from the original Fisher metric on the tangent bundle.  What we have shown are 
listed below. 
\begin{enumerate}
\item Based on a correspondence between cotangent vectors and random variables, 
the Fisher \red{co-metric} is defined via the variance/covariance in a natural way  (Section~\ref{sec_Fisher_co-metric}). 
\item The Cram\'{e}r-Rao inequality is trivialized by considering the Fisher co-metric 
(Section~\ref{sec_submanifold}). 
\item The role of the m-connection as a connection on the cotangent bundle is  
important in considering the achievability condition for the Cram\'{e}r-Rao inequality 
(Section~\ref{sec_em_connections}). 
\item The monotonicity of the Fisher metric is equivalently translated into the monotonicity of 
the Fisher co-metric and that of the variance (Section~\ref{sec_monotonicity}). 
\item The invariance of the Fisher metric and that of the Fisher co-metric are not logically equivalent, 
and a new \v{C}encov-type theorem for characterizing the Fisher co-metric by the invariance is established, 
which can also be regarded as a theorem for characterizing the variance/covariance (Section~\ref{sec_invariance}). 
\item The notion of strong invariance is introduced, which combines  
the invariance of the Fisher metric and that of the Fisher co-metric (Section~\ref{sec_strong_invariance}). 
\item The weak invariance of the $\alpha$-connections is expressed in a 
formulation similar to the strong invariance \red{of the Fisher metric/co-metric}  
 (Section~\ref{sec_invariance_connection}). 
\end{enumerate}

\medskip

It should be noted that, although this paper emphasizes the importance of the Fisher co-metric, this does not 
diminish the importance of the Fisher metric at all.  Apart from the importance 
as a metric on the tangent bundle itself, which is essential for geometry of statistical manifolds, 
 we should not forget that the Fisher information matrix (i.e.\ the components of the Fisher metric) 
is of primary importance as a practical tool to compute the Fisher co-metric. 
Even if we know that $g_{M}^{ij} (p) = g_{M, p} ((d\xi^i)_p, (d\xi^j)_p)$ in 
\eqref{g^ij_M_p} can be defined by 
\eqref{eq1_prop_norm_submanifold} 
and \eqref{eq2_prop_norm_submanifold}  
and that understanding $g_{M}^{ij} (p)$ 
in this way is important for conceptual understanding of the 
Cram\'{e}r-Rao inequality, this does not tell us a method to 
compute $g_{M}^{ij} (p)$ for a given statistical model $(M, \xi)$ 
better than computing the inverse of the Fisher information 
matrix $G_M (p) := [g_{M, ij} (p) ]$ in general. 

Finally, we  note that some of the results obtained here can be extended to the quantum case in several directions, 
which will be discussed in a forthcoming paper. 


\section*{Acknowledgments}
This work was partly supported by JSPS KAKENHI Grant Numbers 23H05492.

\renewcommand{\thesection}{A} 
 \section*{Appendix}
 \label{sec_appendix}
\renewcommand{\thesubsection}{\thesection\arabic{subsection}}
\setcounter{subsection}{0}

\subsection{Proof of  Theorem~\ref{theorem_invariance_covariance}}
\label{subsec_proof_theorem_invariance_covariance}

As mentioned in Remark~\ref{remark_theorem_invariance_L2}, we show the 
equivalence (i)${}'$ $\Leftrightarrow$ (ii-2), with expressing (i)${}'$ as
\smallskip 
\begin{enumerate}
\item[(i)${}'$]
$
\some c_1, \some c_2\in \bR, \, \any n,  \any p\in \cP_n, \;
\gamma_{n, p} = c_1\, \inprod{\cdot}{\cdot}_p + c_2\, \expect{\cdot}_p \expect{\cdot}_p.
$
\end{enumerate}
\smallskip
The implication  (i)${}'$ $\Rightarrow$ (ii-2) is obvious. To show the converse, we assume (ii-2) and will 
derive (i)${}'$ by several steps.
The uniform distribution on $\Omega_n$ is denoted by $u_n$ in the sequel. 
\medskip
\begin{enumerate}
\item[(a)] $\any n, \,\some c_{1,n}, \some c_{2, n} \in \bR, \; 
\gamma_{n, u_n} = c_{1,n} \,  \inprod{\cdot}{\cdot}_{u_n} + c_{2,n} \, \expect{\cdot}_{u_n} \expect{\cdot}_{u_n}.
$
\begin{proof}
Fix $n\geq 2$ arbitrarily, and let  $m=n$, and $F$ be a permutation on $\Omega_n$. 
Define $e_{n, i}\in \bR^{\Omega_n}$ for $i\in \Omega_n$ by $e_{n, i} (j) = \delta_{i, j}$. 
Then ${u_n}^F = u_n$, and condition (ii-2) claims that 
\ba
\any i, \any j\in \Omega_n, \; \gamma_{n, u_n} (e_{n, i}, e_{n, j}) & 
=  \gamma_{n, u_n}(e_{n, i}\circ F, e_{n, j}\circ F)  
\noret & = \gamma_{n, u_n} (e_{n, F^{-1}(i)}, e_{n, F^{-1}(j)}).
\nonumber
\end{align}
Since this holds for any permutation $F$, there exists $a_n, b_n \in \bR$ such that
\be
 \any i, \any j\in \Omega_n, \;\; 
 \gamma_{n, u_n} (e_{n, i}, e_{n, j}) = a_n \delta_{i, j} + b_n.
 \nonumber
\ee
Comparing this with
\be
\inprod{e_{n, i}}{e_{n, j}}_{u_n} = \frac{1}{n}  \delta_{i, j} \quad\text{and}\quad
\expect{e_{n, i}}_{u_n} \expect{e_{n, j}}_{u_n} = \frac{1}{n^2}, 
\nonumber
\ee
we have the assertion of (a) with letting $c_{1,n} := n a_n$ and $c_{2,n} := n^2 b_n$.
\end{proof}
\item[(b)] $\some c_1, \some c_2\in \bR, \, \any n, \; 
\gamma_{n, u_n} = c_1\, \inprod{\cdot}{\cdot}_{u_n} + c_2\, \expect{\cdot}_{u_n} \expect{\cdot}_{u_n}$. 
\begin{proof}
It suffices to show that the constants $\{c_{1,n}\}$ and $\{c_{2,n}\}$ 
 in (a) satisfy $c_{1, n} = c_{1, m}$ and $c_{2, n} = c_{2, m}$ for every $n, m \geq 2$. 
 Let $F: \Omega_{mn} \rightarrow \Omega_n$ be defined by
\be F(i) = k \quad \text{if} \quad (k-1)\, m + 1 \leq i \leq km, \;\; k\in \bN.
\nonumber
\ee
Then we have $u_{mn}^F = u_n$, and (ii-2) yields 
\be
\any A, \any B\in \bR^{\Omega_n}, \;\; \gamma_{n, u_n} (A, B) = \gamma_{mn, u_{mn}} (A\circ F, B\circ F).
\nonumber
\ee
Due to (a), the LHS of the above equation is represented as
\be
 \gamma_{n, u_n} (A, B) = c_{1, n} \inprod{A}{B}_{u_n} + c_{2, n} \expect{A}_{u_n}  \expect{B}_{u_n}, 
 \nonumber
\ee
while the RHS is represented as
\ba
 \gamma_{mn, u_{mn}} (A\circ F, B\circ F) & = 
  c_{1, mn} \inprod{A\circ F}{B\circ F}_{u_{mn}} + c_{2, mn} \expect{A\circ F}_{u_{mn}}  \expect{B\circ F}_{u_{mn}}
  \noret 
 & = c_{1, mn} \inprod{A}{B}_{u_{n}} + c_{2, mn} \expect{A}_{u_{n}}  \expect{B}_{u_{n}}.
 \nonumber
\end{align}
Thus we have 
$c_{1,n} = c_{1, mn}$ and $c_{2,n} = c_{2, mn}$. 
 Similarly, we can show that $c_{1,m} = c_{1, mn}$ and $c_{2,m} = c_{2, mn}$. 
 Hence we obtain $c_{1, n} = c_{1, m}$ and $c_{2, n} = c_{2, m}$.
\end{proof}
\item[(c)] For any $n\geq 2$ and any $p\in \cP_n$ whose components $\{p(i)\}_{i\in \Omega_n}$ are all rational numbers, 
we have $\gamma_{n, p}  = c_1\, \inprod{\cdot}{\cdot}_p + c_2 \expect{\cdot}_p \expect{\cdot}_p $, 
where $c_1$ and $c_2$ are the constants in (b). 
\begin{proof}
Suppose that the components of $p$ are represented as
$p(i) = k_i / m$ by natural numbers $m, k_1, \ldots \red{,}\, k_n$ satisfying $\sum_{i=1}^n k_i = m$.  
Let $\Omega_m = \bigsqcup_{i=1}^n K_i$ be a partition of $\Omega_m$ such that 
$\vert K_i \vert = k_i$, and define $F : \Omega_m \rightarrow \Omega_n$ by
\be
F(j) = i \quad \text{if}\quad j \in K_i.
\nonumber
\ee
Then we have $u_m^F = p$. 
Letting $\gamma'_{n, p} := c_1\, \inprod{\cdot}{\cdot}_p + c_2 \expect{\cdot}_p \expect{\cdot}_p$ and 
noting that both $\{\gamma_{n, p}\}$ and $\{\gamma'_{n, p}\}$ satisfy condition (ii-2),  we obtain  
\ba
\any A, B\in \bR^{\Omega_n}, \;\; 
\gamma_{n, p} (A, B) & = \gamma_{m, u_m} (A\circ F, B\circ F) 
\noret 
&=\gamma'_{m, u_m} (A\circ F, B\circ F)  
\noret 
&= \gamma'_{n, p} (A, B), 
\nonumber
\end{align}
where the second equality follows from (b). 
\end{proof}
\item[(d)] (i)${}'$ is derived from (c) and the continuity of $p\mapsto \gamma_{n, p}$.  
\end{enumerate}
\hfill (QED)

\subsection{An exact form of Proposition~\ref{proposition_characterization_co-metric_strong_invariance}} 
\label{subsec_exact_proposition_characterization_co-metric_strong_invariance}

We present a proposition of \v{C}encov-type that can be regarded as an exact version of Proposition~\ref{proposition_characterization_co-metric_strong_invariance}, and 
show that the proposition follows from Theorem~\ref{theorem_invariance_metric} (\v{C}encov's theorem). 

\begin{proposition}
\label{prop_strong_invariance_co-metric}
For $n=2, 3, \ldots$, let $\Omega_n:= \{1, 2, \ldots , n\}$ and $\cP_n:= \cP (\Omega_n)$, and let 
$g_n$ be the Fisher co-metric on $\cP_n$.  Suppose that we are given a sequence $\{h_n\}_{n=2}^\infty$, where 
$h_n$ is a contravariant tensor field of degree 2 on $\cP_n$ which  continuously maps each point $p\in \cP_n$ to a 
non-degenerate and symmetric bilinear form $h_{n, p} : \cotan{p} (\cP_n)^2 \rightarrow \bR$.  Then the following two conditions are equivalent.
\begin{enumerate}
\item[(i)] $\some c\in \bR\setminus\{0\}, \, \any n, \;\; h_n = c g_n$.
\item[(ii)] For any $m\leq n$ and any pair $(\Phi, \Psi)$ of 
Markov embedding $\Phi :\cP_m\rightarrow \cP_n$ and Markov co-embedding $\Psi : \cP_n \rightarrow \cP_m$ satisfying 
$\Psi \circ \Phi = \id_{\cP_m}$,  it holds that
\ba
\any p\in \cP_m,\, & \any \alpha_{p} \in \cotan{p} (\cP_m), 
\any \beta_{\Phi(p)}  \in \cotan{\Phi(p)} (\cP_n), 
\noret & 
h_{m, p} (\alpha_{p}, \Phi^* (\beta_{\Phi (p)})) = h_{n, \Phi (p)} (\Psi^* (\alpha_p), \beta_{\Phi (p)}).
\label{(ii)_prop_strong_invariance_co-metric_appendix}
\end{align}
\end{enumerate} 
\end{proposition}

\begin{proof} 
Since the Fisher co-metric satisfies \eqref{strong_invariance_Fisher_cotangent_2}, 
we have (i) $\Rightarrow$ (ii). 
To prove the converse, assume that $\{h_n\}_{n=2}^\infty$ satisfies (ii). 
Due to the assumption that $h_{n, p}$ is non-degenerate for every $n$ and $p\in \cP_n$, 
$h_{n, p}$ induces the correspondence 
$\stackrel{h_{n, p}}{\longleftrightarrow}$ between $\tan{p} (\cP_n)$ and $\cotan{p} (\cP_n)$ 
together with the bilinear form $h_{n, p} : \tan{p} (\cP_n)^2 \rightarrow \bR$ 
as in the case of inner products by
\ba
X_p \stackrel{h_{n, p}}{\longleftrightarrow} \alpha_p \; & \Leftrightarrow \; 
\any \beta_p \in \cotan{p} (\cP_n), \; \beta_p (X_p) = h_{n, p} (\alpha_p, \beta_p)
\nonumber
\end{align}
and 
\be
X_p \stackrel{h_{n, p}}{\longleftrightarrow} \alpha_p 
\;\; \text{and} \;\;
Y_p \stackrel{h_{n, p}}{\longleftrightarrow} \beta_p 
\; \Rightarrow \;
h_{n, p} (X_p, Y_p) = h_{n, p} (\alpha_p, \beta_p). 
\nonumber
\ee
Then, similar to the equivalence between \eqref{strong_invariance_Fisher_tangent_1} and \eqref{strong_invariance_Fisher_cotangent_2}, we see that 
\eqref{(ii)_prop_strong_invariance_co-metric_appendix} is equivalent to 
\ba
\any p\in \cP_m,\, & \any X_{p} \in \tan{p} (\cP_m), 
\any Y_{\Phi(p)}  \in \tan{\Phi(p)} (\cP_n), 
\noret & 
h_{m, p} (X_{p}, \Psi_* (Y_{\Phi (p)})) = h_{n, \Phi (p)} (\Phi_* (X_p),  Y_{\Phi (p)}).
\label{(ii)_prop_strong_invariance_co-metric_equiv}
\end{align}
Indeed, if 
\be
X_p \stackrel{h_{m, p}}{\longleftrightarrow} \alpha_p \;\; \text{and} \;\;
Y_{q} \stackrel{h_{n, q}}{\longleftrightarrow} \beta_q, \quad\text{where}\quad 
q := \Phi (p), 
\nonumber
\ee
then
\be
h_{m, p} (\alpha_p, \Phi^* (\beta_q)) = \Phi^* (\beta_q) (X_p) = \beta_q (\Phi_* (X_p)) 
= h_{n, q} (\Phi_* (X_p), Y_q)
\nonumber
\ee
and
\be
h_{n, q} (\Psi^* (\alpha_p), \beta_{q}) = \Psi^* (\alpha_p) (Y_q) = 
\alpha_p ( \Psi_* (Y_q)) = h_{m, p} (X_p, \Psi_* (Y_q)), 
\nonumber
\ee
which shows the equivalence of 
\eqref{(ii)_prop_strong_invariance_co-metric_appendix}
and \eqref{(ii)_prop_strong_invariance_co-metric_equiv}. 
Due to $\Psi_* \circ \Phi_* = \id_{\tan{p}  (\cP_m)}$, Eq.~\eqref{(ii)_prop_strong_invariance_co-metric_equiv} implies \eqref{(ii)_theorem_invariance_metric} in 
(ii) of Theorem~\ref{theorem_invariance_metric}. 
Thus we have the following train of implications:
\ba
\text{(ii)} \; & \Rightarrow \; \text{(ii) of Theorem~\ref{theorem_invariance_metric}} 
\Rightarrow \; \text{(i) of Theorem~\ref{theorem_invariance_metric}} 
\Rightarrow \; \text{(i)}.
\nonumber
\end{align}
\end{proof}




\end{document}